\documentclass{scrartcl}
\usepackage{amsfonts}
\usepackage{amsmath}
\usepackage{amssymb}
\usepackage{amsthm}
\usepackage{mathtools}
\usepackage{enumerate}
\usepackage{dsfont}
\usepackage{graphicx}
\usepackage{nicefrac}
\usepackage{color}
\usepackage[all]{xy}
\usepackage{hyperref}
\usepackage{bbm}
\usepackage{tikz}
\usepackage{tikz-cd}
\usepackage{cancel}
\usepackage{xcolor}
\usepackage{booktabs}
\usepackage{multirow}
\usepackage{subcaption}
\usetikzlibrary{positioning}
\usetikzlibrary{arrows}
\usepackage{MnSymbol}
\usepackage[normalem]{ulem}
\usepackage{comment}

\definecolor{db}{RGB}{66, 126, 147}

\newtheorem{theorem}{Theorem}

\newtheorem{corollary}{Corollary}
\newtheorem{lemma}{Lemma}
\newtheorem{definition}{Definition}
\newtheorem*{question*}{Question}

\newcommand{\zz}{\mathbb Z}

\newcommand{\R}{\mathbb R}
\newcommand{\C}{\mathbb C}

\newcommand{\tr}{\mathrm{tr}}
\newcommand{\one}{\mathbbm{1}}

\newcommand{\Sym}{\mathrm{Sym}}

\newcommand\mc[1]{\mathcal{#1}}

\newcommand{\spec}{\mathrm{sp}}

\newcommand{\im}{\mathrm{im}}
\newcommand{\sa}{\mathrm{sa}}

\newcommand{\fl}{\mathfrak{l}}

\newcommand{\cB}{{\mc{B}}}

\newcommand{\cO}{{\mc{O}}}
\newcommand{\cP}{\mc{P}}
\newcommand{\cS}{\mc{S}}

\newcommand{\cC}{{\mc{C}}}

\newcommand{\CH}{{\cC(\cH)}}
\newcommand{\CHm}{{\cC_\mathrm{max}(\cH)}}
\newcommand{\CSpin}{{\cC(\cS)}}
\newcommand{\CSm}{{\cC_\mathrm{max}(\cS)}}
\newcommand{\PS}{{\cP(\cS)}}
\newcommand{\PC}{{\cP(C)}}

\newcommand{\cH}{\mc{H}}

\newcommand{\LH}{\mc{L}(\cH)}
\newcommand{\LHsa}{\LH_\sa}

\newcommand{\SH}{\mc{S}(\cH)}

\newcommand{\PH}{{\mc{P}(\cH)}}

\newcommand{\tC}{{\widetilde{C}}}
\newcommand\tO{{\widetilde{O}}}

\newcommand{\tp}{{\widetilde{p}}}

\newcommand\id{{\mathrm{id}}}

\newcommand\ra{\rightarrow}

\newcommand\Llra{\Longleftrightarrow}

\begin{document}
\title{An algebraic characterisation of Kochen-Specker contextuality}
\author{Markus Frembs\\\small{Institut f\"ur Theoretische Physik, Leibniz Universit\"at Hannover,}\\[-0.2cm]\small{Appelstraße 2, 30167 Hannover, Germany}\\\small{Centre for Quantum Dynamics, Griffith University,}\\[-0.2cm]\small{Yugambeh Country, Gold Coast, QLD 4222, Australia}}
\date{\large\today}
\maketitle

\begin{abstract}
    Contextuality is a key distinguishing feature between classical and quantum physics. It expresses a fundamental obstruction to describing quantum theory using classical concepts. In turn, understood as a resource for quantum computation, it is expected to hold the key to quantum advantage. Yet, despite its long recognised importance in quantum foundations and, more recently, in quantum computation, the structural essence of contextuality has remained somewhat elusive - different frameworks address different aspects of the phenomenon, yet their precise relationship often remains unclear. This issue already looms large at the level of the Bell-Kochen-Specker theorem: while traditional proofs proceed by showing the nonexistence of valuations, the notion of state-independent contextuality in the marginal approach allows to prove contextuality from seemingly weaker assumptions. In the light of this, and at the absence of a unified mathematical framework for Kochen-Specker contextuality, the original algebraic approach has been widely abandoned, in favour of the study of contextual correlations.
    
    Here, we reinstate the algebraic perspective on contextuality. Concretely, by building on the novel concept of \emph{context connections}, we reformulate the algebraic relations between observables originally postulated by Kochen and Specker, and we explicitly demonstrate their consistency with the notion of state-independent contextuality. In the present paper, we focus on the new conceptual ideas and discuss them in the concrete setting of spin-$1$ observables, specifically those in the example of Ref.~\cite{YuOh2012}; in a companion paper \cite{Frembs2024b}, we generalise these ideas, obtain a complete characterisation of Kochen-Specker contextuality and provide a detailed comparison with the related notions of contextuality in the marginal and graph-theoretic approach.
\end{abstract}

\section{Introduction}\label{sec: motivation}

The Bell-Kochen-Specker (BKS) theorem \cite{Bell1966,KochenSpecker1967} demonstrates that quantum theory does not admit a classical description, unless such a description is both nonlocal and contextual. More precisely, there exists no embedding $\epsilon:\LHsa\ra L_\infty(\Lambda)$ of quantum observables (represented by self-adjoint operators $\LHsa$ on a Hilbert space $\cH$) into the space of (measurable) functions on some classical state space $\Lambda$ (`hidden variables'), such that algebraic relations between commuting observables are preserved. Following Kochen and Specker's original argument, various proofs of the BKS theorem construct so-called \emph{Kochen-Specker sets}, that is, sets of rank-$1$ projections which do not admit a valuation (for details, see Sec.~\ref{sec: BKS theorem}). 

An alternative proof strategy is to show the violation of \emph{noncontextuality inequalities} by quantum mechanics (for details, see Sec.~\ref{sec: noncontextuality inequalities}). As it turns out, there exist noncontextuality inequalities that are violated by every quantum state, yet which do not constitute Kochen-Specker sets \cite{YuOh2012}.\footnote{Note also that this is unlike the case of Bell inequality violations, which are necessarily state-dependent.} This came as a surprise, as it indicated a potential discrepancy between the original characterisation of contextuality by Kochen and Specker on the one side, and the modern analysis of noncontextual correlations on the other \cite{BudroniEtAl2022}. In any case, it expresses the fact that the full extent of the constraints on classical embeddings $\epsilon$ remains poorly understood. Given the importance of contextuality in quantum foundations, and as a resource in quantum computation, this is an unsatisfactory state of affairs. The aim of this and a companion paper is to remedy this situation. The present paper focuses on the necessary new concepts and shows how they resolve the apparent discrepancy with the notion of state-independent contextuality in the marginal approach; in a companion paper \cite{Frembs2024b}, we lift these to a complete algebraic characterisation of Kochen-Specker noncontextuality, and provide a full comparison with various related notions of contextuality in the marginal approach.

The present paper is structured as follows. In Sec.~\ref{sec: two views}, we provide some minimal technical background (Sec.~\ref{sec: contexts in order}) in order to review both the original ``algebraic approach" to contextuality, as initiated by Kochen and Specker (Sec.~\ref{sec: KS contextuality}), the Bell-Kochen-Specker theorem (Sec.~\ref{sec: BKS theorem}) and the ``marginal approach" which studies noncontextual correlations (Sec.~\ref{sec: noncontextuality inequalities}). Sec.~\ref{sec: fundamental discrepancy} contrasts these two approaches and stresses that the algebraic one lacks an adequate counterpart to the characterisation of Kochen-Specker contextuality in terms of the state-independent violation of noncontextuality inequalities in the marginal approach. Sec.~\ref{sec: wielding the unwieldy} contains our new contributions. We first provide a physical intuition for the insufficiency of standard arguments in proofs of the BKS theorem based on Kochen-Specker sets in Sec.~\ref{sec: BKS theorem}. This intuition is then formalised in Sec.~\ref{sec: context connection}, where we introduce the key new concept and corresponding mathematical tool of this work: \emph{context connections}. Sec.~\ref{sec: CNC} contains our main result (Thm.~\ref{thm: CNC}), which derives new, inherently geometric constraints from Kochen-Specker noncontextuality. Moreover, in Sec.~\ref{sec: resolving the discrepancy} we show explicitly that these constraints subsume proofs of the BKS theorem based on state-independent violations of noncontextuality inequalities. Sec.~\ref{sec: discussion} concludes.

\section{Two perspectives on Kochen-Specker contextuality}\label{sec: two views}

\subsection{Quantum observables, contexts and context order}\label{sec: contexts in order}

In quantum theory, sharp (or ideal) obversables\footnote{Our restriction to sharp observables follows the original treatment by Kochen and Specker \cite{KochenSpecker1967} and that of Ref.~\cite{AbramskyBrandenburger2011,DoeringFrembs2019a,XuCabello2019}, and we adopt it here for ease of comparison. Still, we remark that the essential notion of compatibility of observables (see below) has since been extended to unsharp observables in the form of positive operator-valued measures \cite{BuschLahtiPellonpaa_QuantumMeasurement,Kunjwal2014,ChiribellaYuan2014,HeinosaariMiyaderaZiman2016}. An analysis of Kochen-Specker contextuality in terms of unsharp observables is thus possible, and we defer it to future work.} are represented by self-adjoint or Hermitian operators $\LHsa$ on some Hilbert space $\cH$.\footnote{We will use the terms (quantum) observable and Hermitian (self-adjoint) operator interchangeably.}\footnote{Kochen-Specker contextuality is not tied to quantum theory, and many frameworks in the literature do not restrict observables to be quantum \cite{KochenSpecker1967,AbramskyBrandenburger2011,CabelloSeveriniWinter2014}. Indeed, Kochen and Specker consider partial algebras - of which $\LHsa$ is a rather special case \cite{KochenSpecker1967}. Our formalism is not tied to the quantum setting either, however, for clarity of exposition, in this paper we only discuss the finite-dimensional quantum case, leaving the general case to a companion paper \cite{Frembs2024b}.} The Heisenberg uncertainty relations assert that not all quantum observables can be measured jointly (in every state) \cite{Heisenberg1927}. Mathematically, this is a consequence of the noncommutativity of quantum observables. In turn, we call two Hermitian operators $O,O'\in\LHsa$ \emph{jointly measurable} or \emph{compatible} if they commute $[O,O'] = 0$. The following alternative characterisation of commutativity will be useful. Let $O\in\LHsa$ be a self-adjoint operator with spectral decomposition $O=\sum_k\lambda^O_k p_k$ for $\lambda^O_k\in\R$ and $p_k\in\PH$ a projection. We call $\mathrm{sp}(O) = \{\lambda^O_k\}_k$ the \emph{spectrum of $O$}, and $O$ \emph{non-degenerate} if $|\mathrm{sp}(O)| = d = \dim(\cH)$; otherwise, $O$ is called \emph{degenerate}. Given any function $h:\R\ra\R$, we obtain another self-adjoint operator $h(O)\in\LHsa$ by setting $h(O) = \sum_k h(\lambda^O_k)p_k$.\footnote{Since we restrict to finite dimensions, the spectra of all observables are discrete. We therefore need not further restrict the space of functions $h:\R\ra\R$, and will leave them arbitrary.} Now, let $O',O'' \in \LHsa$, then (see e.g. Thm.~11.3 in Ref.~\cite{BuschLahtiPellonpaa_QuantumMeasurement}):
\begin{align}\label{eq: commuting observables}
    [O',O''] = 0 \quad \Longleftrightarrow \quad \exists O \in \LHsa, g',g'': \R \ra \R\ \mathrm{s.t.}\ O' = g'(O), O''=g''(O)\; .
\end{align}
We group observables into collections of jointly measurable ones, by defining a \emph{context $C\subset\LHsa$} as a commutative subalgebra of $\LHsa$.\footnote{This idea is the starting point of the topos approach \cite{IshamButterfieldI,IshamButterfieldII,IshamButterfieldIII,IshamButterfieldIV,IshamDoeringI,IshamDoeringII,IshamDoeringIII,IshamDoeringIV,HLSBohrification2009,HLS2009,HLS2010,DoeringFrembs2019a,FrembsDoering2022a}. Here, define contexts to be real subalgebras, instead of considering their complexifications as is often done.} Note that the observables $O \in C$ in a context $C\subset\LHsa$ all have a spectral decomposition with respect to the same projections $\cP(C)=\{p\in\PH\mid\forall O \in C:\ \tr[Op] \in \mathrm{sp}(O)\}$. Every commutative subalgebra $C\subset\LHsa$ is therefore generated by a single self-adjoint operator $O\in\LHsa$,
\begin{align}\label{eq: generating observables}
    C = C(O) := \langle h(O) \mid h: \R \ra \R\rangle\; .
\end{align}
Note that there are generally many observables generating the same context, that is, $C(O)=C(O')$ does not imply $O=O'$. Nevertheless, if $O$ is non-degenerate, then $C(O)=C(O')$ implies that $O'$, too, is a non-degenerate observable. Clearly, the set of contexts generated by non-degenerate observables is maximal, that is, $C(O) \subset C$ for $O$ non-degenerate and $C\subset\LHsa$ a context implies $C(O) = C$. We denote the set of maximal contexts by $\CHm$. In turn, $C\subset\LHsa$ is non-maximal if $C=C(O)$ with $O\in\LHsa$ a degenerate observable. The set of all contexts, ordered by inclusion,
\begin{align}\label{eq: context category}
    \CH := (\{C\subset\LHsa \mid C\ \mathrm{commutative}\},\subset)\; ,
\end{align}
is called the \emph{partial order of contexts} or \emph{context category $\CH$}. See Fig.~\ref{fig: context category} (and Fig.~\ref{fig: spin context category}) for a sketch of the partial order of contexts $\cC(\C^3)$ of a spin-$1$ quantum system.

\begin{figure}
    \centering
    \scalebox{1.3}{\begin{tikzpicture}[node distance=3cm, every node/.style={scale=0.67}]
    \node(Vp1)                        {$\mathbb{R}p_{-1}+\mathbb{R}(\mathds{1}-p_{-1})$};
    \node(Va2c)   [above=1.25cm of Vp1]        {};
    \node(V1bc)   [left=0.5cm of Va2c]        {};
    \node(Vp2)   [right of=Vp1]        {$\mathbb{R}p_1+\mathbb{R}(\mathds{1}-p_1)$};
    \node(Vp3)   [right of=Vp2]        {$\mathbb{R}p_0+\mathbb{R}(\mathds{1}-p_0)$};
    \node(Vab3)   [above=1.25cm of Vp3]        {};
    \node(V123)   [above=1.25cm of Vp2]        {$\mathbb{R}p_{-1}+\mathbb{R}p_0+\mathbb{R}p_1$};
    
    \node(Vp'2)   [right of=Vp3]      {$\mathbb{R}p'_{-1}+\mathbb{R}(\mathds{1}-p'_{-1})$};
    \node(Vp'1)   [right of=Vp'2]       {$\mathbb{R}p'_1+\mathbb{R}(\mathds{1}-p'_1)$};
    \node(Va2'c)   [above=1.25cm of Vp'1]        {};
    \node(V1'bc)   [right=0.5cm of Va2'c]        {};
    \node(V1'2'3)   [above= 1.25cm of Vp'2]      {$\mathbb{R}p'_{-1}+\mathbb{R}p_0+\mathbb{R}p'_1$};
    
    \node(Vone)  [below=1.25cm of Vp3]  {$\mathbb{R}\mathbbm{\mathds{1}}$};
    
    \node(3DotsUpLeft)     [left=0.3cm of V1bc]    {$\cdots$};
    \node(3DotsUpRight)     [right=0.3cm of V1'bc]    {$\cdots$};
    \node(3DotsDownLeft)   [left=0.3cm of Vp1]    {$\cdots$};
    \node(3DotsDownRight)   [right=0.3cm of Vp'1]    {$\cdots$};

    \draw [->](Vone) -- (Vp1);
    \draw [->](Vone) -- (Vp2);
    \draw [->](Vone) -- (Vp3);
    \draw [->](Vone) -- (Vp'2);
    \draw [->](Vone) -- (Vp'1);
    
    \draw [->](Vp1) -- (V123);
    \draw [->,dashed](Vp1) -- (V1bc);
    \draw [->](Vp2) -- (V123);
    \draw [->,dashed](Vp2) -- (Va2c);

    \draw [->](Vp3) -- (V123);
    \draw [->](Vp3) -- (V1'2'3);
    \draw [->,dashed](Vp3) -- (Vab3);
    
    \draw [->](Vp'2) -- (V1'2'3);
    \draw [->,dashed](Vp'2) -- (Va2'c);
    \draw [->](Vp'1) -- (V1'2'3);
    \draw [->,dashed](Vp'1) -- (V1'bc);
\end{tikzpicture}}
    \caption{Sketch of the partial order of contexts $\CH$ of a spin-$1$ system (for $\dim(\cH) = 3$). The least element of $\CH$ is generated by the identity $C(\mathbbm{1})=\R\mathbbm{1}$. Every maximal context $C\in\CHm$ is generated by three orthogonal rank-$1$ projections, e.g. the projections onto the eigenspaces of a spin-$1$ observable $S=\sum_{s=0,\pm 1} sp_s$ define the context $C(S)=\R p_{-1}+\R p_0 + \R p_1$. Finally, every non-maximal context $C(\mathbbm{1})\neq C\in\CH$ is generated by a single rank-$1$ projection and its complement, e.g. $C(S^2)=\R p_0 + \R(\mathbbm{1}-p_0)$.}
    \label{fig: context category}
\end{figure}
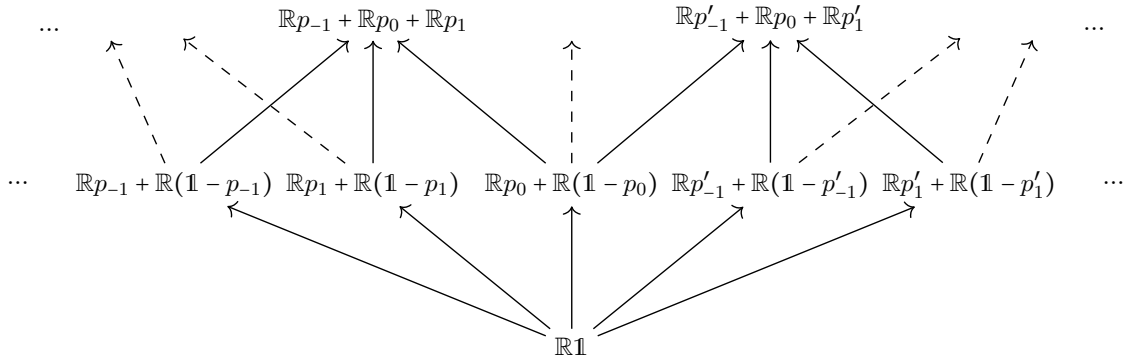

\subsection{Kochen-Specker (non)contextuality}\label{sec: KS contextuality}

A central question in the foundations of quantum mechanics is whether quantum theory is ``incomplete" and will eventually give way to a more fundamental description in terms of classical ``hidden variables" \cite{EPR1935,Marage1999}. If the counterintuitive predictions of quantum mechanics could be rooted in a classical theory, it would help resolve at least some of the issues regarding the interpretation of quantum theory. In turn, and from a modern perspective, the promise of quantum computation relies on the assumption that quantum mechanics escapes a classical description, and thus also classical limitations on computational power. The study of contextuality as a resource in quantum computing is an active and growing field of research \cite{Raussendorf2013,HowardEtAl2014,FrembsRobertsBartlett2018,FrembsRobertsCampbellBartlett2023}.\\

\textbf{Classical embeddings.} For $\LHsa$ to admit a description in terms of a classical state space (``hidden variables") $\Lambda$, requires a \emph{(classical) embedding $\epsilon: \LHsa \ra L_\infty(\Lambda)$} from the space of quantum observables $\LHsa$ into the algebra of measurable functions on $\Lambda$.\footnote{From the perspective of quantum computation, the existence of a classical embedding $\epsilon:\cO\ra L_\infty(\Lambda)$ for a subset of operators $\cO\subset\LHsa$ is of interest since it suggests that the respective quantum sub-theory can be efficiently classically simulated; in contrast, the nonexistence of such embeddings indicates that the sub-theory is sufficient to achieve ``quantum advantage".} Usually, the state or phase space of a classical theory is a symplectic manifold, hence, equipped with both topological and geometric structure \cite{Arnold2013}. However, for our purposes it will be sufficient that $\Lambda$ can represent classical observables in the form of measurable functions, hence, we will require $\Lambda$ to be a measurable space only.\footnote{Recall that a measurable space $(\Lambda,\sigma)$ is a set $\Lambda$ equipped with a $\sigma$-algebra of measurable subsets of $\Lambda$, which is closed under complements, countable unions and countable intersections. A measure space further specifies a measure $\mu$ on $(\Lambda,\sigma)$, satisfying $\mu(S) \geq 0$ for all $S \in \sigma$ and $\mu(\bigcup_{k=1}^\infty S_k) = \sum_{k=1}^\infty \mu(S_k)$ for countable collections $\{S_k\}_{k=1}^\infty$ of pairwise disjoint sets. If $\Lambda$ is equipped with a topology, then a canonical choice of $\sigma$-algebra of $\Lambda$ is given by the Borel $\sigma$-algebra generated by the open sets in $\Lambda$.}

What constraints should a classical embedding $\epsilon$ of $\LHsa$ satisfy? A natural requirement is that it should reproduce quantum mechanical expectation values, that is, for all $O\in\LHsa$ and $\rho\in\SH$ there should exist a measure $\mu_\rho \in L_1(\Lambda)$ such that
\begin{align}\label{eq: preserving expectation values}
    \mathbb{E}_\rho(O) = \tr[\rho O] = \int_\Lambda d\lambda\ \mu_\rho(\lambda) \epsilon(O)(\lambda)\; .
\end{align}
Is Eq.~(\ref{eq: preserving expectation values}) sufficient to exclude hidden variables for quantum theory? Clearly not. An explicit example for a classical representation is Bohm's pilot wave theory \cite{Bohm1952}. Is this the unique such representation? If not, how restrictive is Eq.~(\ref{eq: preserving expectation values})? As it turns out, not very. Indeed, Ref.~\cite{BeltramettiBugajski1995} describes a general construction for building a classical representation for quantum theory satisfying Eq.~(\ref{eq: preserving expectation values}); and already Kochen and Specker point out a trivial construction of a classical representation satisfying Eq.~(\ref{eq: preserving expectation values}) that simply defines the state space $\Lambda = \prod_{O\in\LHsa} \Sigma_O$ as the product of individual state spaces $\Sigma_O = \{\lambda:O\ra\spec(O)\}$ over all observables $O\in\LHsa$ \cite{KochenSpecker1967}.

This highlights that a map $\epsilon:\LHsa\ra L_\infty(\Lambda)$ is per se not required to preserve any algebraic structure in $\LHsa$ - even if it satisfies Eq.~(\ref{eq: preserving expectation values}) (cf. Ref.~\cite{BeltramettiBugajski1995}). In turn, asking $\epsilon$ to preserve (some of) the algebraic structure in $\LHsa$ quickly leads to no-go results, such as von Neumann's, proving that no map $\epsilon$ satisfying Eq.~(\ref{eq: preserving expectation values}) and preserving the linear structure in $\LHsa$ exists \cite{vonNeumann1932}.\footnote{Similarly, the results in Ref.~\cite{Groenewold1946,vanHove1951b} can be interpreted as showing that no classical embedding exists that preserves the full Jordan-algebraic structure on $\LHsa$. Note that this result is usually read in reverse, namely as proving the necessity of a choice of operator ordering for quantisation.} Yet, why should a classical embedding respect the algebraic structure in $\LHsa$? After all, the search for a classical description underlying quantum theory was motivated by the counterintuitive consequences implied by its algebraic apparatus. For instance, if we take the Heisenberg uncertainty relations to be a fundamental limitation on the joint measurability of observables, we cannot operationally infer the linear structure of quantum mechanics, since addition of self-adjoint operators in $\LHsa$ is defined for all, not just jointly measurable observables.\footnote{Grete Hermann was the first to point out this issue in von Neumann's argument against hidden variables \cite{CrullBacciagaluppi2016}; Bell later raised similar concerns \cite{Bell1966,Bell1982}. Still, linearity \emph{does} follow from quasi-linearity, that is, from linearity restricted to commuting operators, as a consequence of Gleason's theorem \cite{Gleason1975}. For a thorough account of this historical episode, see Ref.~\cite{Dieks2017}.}

The crucial insight by Kochen and Specker \cite{Specker1960,KochenSpecker1967} (and independently by Bell \cite{Bell1966}) is that the unavoidable disturbance of a quantum system caused by measurement - an empirical fact for (and, in this view, defining of) incompatible observables - only allows to operationally infer algebraic relations between compatible ones. These algebraic relations are independent of the formalisation, hence, if present in quantum theory should also be preserved under a classical embedding $\epsilon$. Since quantum observables are jointly measurable whenever they commute, we are thus led to impose algebraic constraints in the form of algebraic relations between commuting observables. 

Formally, let $\tO,O\in\LHsa$ be two quantum observables such that $\tO=g(O)$ for some function $g:\R\ra\R$. Then such algebraic relations should be respected by $\epsilon$, that is,
\begin{equation}\label{eq: KSNC}
    \forall O\in\LHsa,\forall g:\R\ra\R: \quad \epsilon(\tO) = \epsilon(g(O)) = g(\epsilon(O)) \; .
\end{equation}
Eq.~(\ref{eq: KSNC}) as a constraint for a (hypothetical) classical embedding of $\LHsa$ was first proposed by Kochen and Specker \cite{KochenSpecker1967}. Notably, Eq.~(\ref{eq: KSNC}) imposes a \emph{noncontextuality constraint} on $\epsilon$: the measurable function $\epsilon(\tO)$ representing the quantum observable $\tO\in\LHsa$ under the embedding $\epsilon$ does not depend on (other jointly measurable observables in) the context $C(O)$. We call a quantum subsystem $\cO\subset\LHsa$ \emph{Kochen-Specker (KS) noncontextual} if it admits a (classical) embedding $\epsilon:\cO\ra L_\infty(\Lambda)$ for some measurable space $\Lambda$ satisfying Eq.~(\ref{eq: KSNC}), and \emph{(KS) contextual} otherwise.\footnote{We thus treat the characterisation of Kochen-Specker contextuality in Eq.~(\ref{eq: KSNC}) independent of Eq.~(\ref{eq: preserving expectation values}).}\\

\textbf{KS contextuality for spin-$1$ observables.} Since quantum observables commute if and only if they are related by Eq.~(\ref{eq: commuting observables}), and since every context is generated by some self-adjoint operator by Eq.~(\ref{eq: generating observables}), it follows that the algebraic constraints in Eq.~(\ref{eq: KSNC}) are fully contained in, and thus a property of the partial order of contexts $\CH$. But how do we infer whether $\LHsa$ is Kochen-Specker (non)contextual from $\CH$? 

Consider a spin-$1$ system with Hilbert space $\cH=\C^3$, then every maximal context $C\in\CHm$ is generated by some spin-$1$ observable, $C=C(S)$ for $S=\sum_{s=0,\pm 1} sp_s$ with $p_s\in\cP_1(\cH)$. However, $S$ is not unique: any other spin-$1$ observable of the form $S_\pi=\sum_{s=0,\pm 1} \pi(s)p_s$ with $\pi\in\Sym(S):=\Sym(\mathrm{\spec(S)})$ a permutation of the eigenvalues in its spectral decomposition, generates the same context, that is, $C(S_\pi)=C(S)$.
\begin{figure}[!htb]
    \centering
    \scalebox{1.3}{\begin{tikzpicture}[node distance=2.75cm, every node/.style={scale=0.7}]
    \node(Vp1)                        {$C(S^2_\kappa)$};
    \node(Va2c)   [above= 1cm of Vp1]        {};
    \node(V1bc)   [left=0.5cm of Va2c]        {};
    \node(Vp2)   [right of=Vp1]        {$C(S^2_\pi)$};
    \node(Vp3)   [right of=Vp2]        {$C(S^2)=C(S'^2)$};
    \node(Vab3)   [above= 1cm of Vp3]        {};
    \node(V123)   [above= 1cm of Vp2]        {$C(S_\kappa)=C(S_\pi)=C(S)$};
    
    \node(Vp'2)   [right of=Vp3]      {$C(S'^2_{\pi'})$};
    \node(Vp'1)   [right of=Vp'2]       {$C(S'^2_{\kappa'})$};
    \node(Va2'c)   [above=1cm of Vp'1]        {};
    \node(V1'bc)   [right=0.5cm of Va2'c]        {};
    \node(V1'2'3)   [above= 1cm of Vp'2]      {$C(S')=C(S'_{\pi'})=C(S'_{\kappa'})$};
    
    \node(Vone)  [below= 1cm of Vp3]  {$C(\mathbbm{1})$};
    
    \node(3DotsUpLeft)     [left=0.3cm of V1bc]    {$\cdots$};
    \node(3DotsUpRight)     [right=0.3cm of V1'bc]    {$\cdots$};
    \node(3DotsDownLeft)   [left=0.3cm of Vp1]    {$\cdots$};
    \node(3DotsDownRight)   [right=0.3cm of Vp'1]    {$\cdots$};

    \draw [->](Vone) -- (Vp1);
    \draw [->](Vone) -- (Vp2);
    \draw [->](Vone) -- (Vp3);
    \draw [->](Vone) -- (Vp'2);
    \draw [->](Vone) -- (Vp'1);
    
    \draw [->](Vp1) -- (V123);
    \draw [->,dashed](Vp1) -- (V1bc);
    \draw [->](Vp2) -- (V123);
    \draw [->,dashed](Vp2) -- (Va2c);

    \draw [->](Vp3) -- (V123);
    \draw [->](Vp3) -- (V1'2'3);
    \draw [->,dashed](Vp3) -- (Vab3);
    
    \draw [->](Vp'2) -- (V1'2'3);
    \draw [->,dashed](Vp'2) -- (Va2'c);
    \draw [->](Vp'1) -- (V1'2'3);
    \draw [->,dashed](Vp'1) -- (V1'bc);
\end{tikzpicture}}
    \caption{Sketch of the partial order of contexts $\CSpin\subset\CH$ (cf. Fig.~\ref{fig: context category}) generated by spin-$1$ observables $S\in\cS\subset\LHsa$. Every maximal context $C\in\CSm$ is generated by six spin-$1$ observables $C=C(S_\pi)$, where $S_\pi=\sum_{s=0,\pm 1} \pi(s)p_s$ for $\pi\in\Sym(S)$ and $p_s\in\cP_1(\cH)$; non-trivial subcontexts are generated by a unique squared spin-$1$ observable.}
    \label{fig: spin context category}
\end{figure}
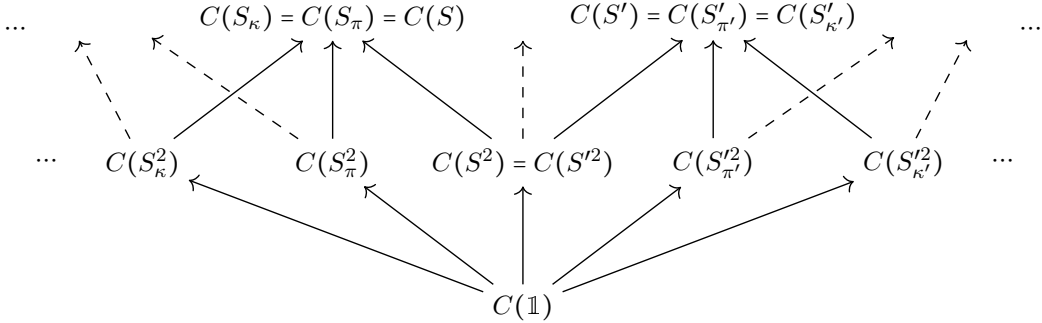

Moreover, the algebraic constraints in Eq.~(\ref{eq: KSNC}) are encoded in the inclusion relations between contexts in $\CH$. For spin-$1$ observables it is sufficient to consider $g:\R\ra\R$, $g(x)=x^2$, since $C(S)\cap C(S')=C(\mathbbm{1})=:C_\mathbbm{1}$ unless $S^2=S'^2$ in which case $C(S)\cap C(S')=C(S^2)=C(S'^2)$ (see Fig.~\ref{fig: spin context category}).\footnote{Moreover, the constant function $g:\R\ra\R$, $g(x)=1$ sends any spin-$1$ observable to the identity, that is, $C(g(S))=C(\mathbbm{1})=C_\mathbbm{1}$ for all $S\in\cS$. (Indeed, this is true for all observables $O\in\LHsa$).} However, note that there exist $\pi\in\Sym(S)$ such that $C(S^2) \neq C(S^2_\pi)$, that is, the squared spin-$1$ observables corresponding to different spin-$1$ observables, which generate the the same maximal context, generally generate different subcontexts. In other words, the algebraic constraints in Eq.~(\ref{eq: KSNC}) impose constraints on how to choose generating spin-$1$ observables across contexts.

\subsection{Valuations and the (Bell-)Kochen-Specker theorem}\label{sec: BKS theorem}

Kochen and Specker deem the constraints in Eq.~(\ref{eq: KSNC}) ``too unwiedly'' \cite{KochenSpecker1967}. Rather than characterising these constraints directly, they observe the following: if a classical embedding $\epsilon: \LHsa \ra L_\infty(\Lambda)$ satisfying Eq.~(\ref{eq: KSNC}) existed, then every (micro)state $\lambda \in \Lambda$ would define a \emph{valuation} (also \emph{truth function}) $v_\lambda: \LHsa\ra\R$ satisfying (i) $v(O) \in \mathrm{sp}(O)$ and (ii) $g(v(O))=v(g(O))$ for all $g:\R\ra\R$ and $O\in\LHsa$ (see Fig.~\ref{fig: valuation from embedding}). Indeed, setting $v_\lambda(O):=\epsilon(O)(\lambda)$, we immediately obtain from Eq.~(\ref{eq: KSNC}),
\begin{align}\label{eq: FUNC}
    v_\lambda(g(O))
    = \epsilon(g(O))(\lambda) 
    = g(\epsilon(O)(\lambda))
    = g(v_\lambda(O))\; .
\end{align}

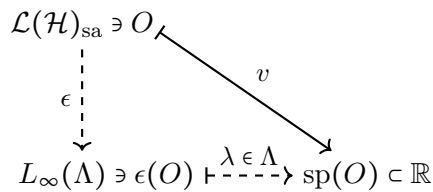
\begin{figure}[!htb]
    \centering
    \begin{tikzpicture}[x=0.05mm,y=0.05mm]
        \node at (0,400)   (O) {$\LHsa \ni O$};
        \node at (175,400)   (OH) {};
        \node at (0,0)   (f) {$\quad \ \ L_\infty(\Lambda)\ni\epsilon(O)$};
        \node at (750,0)   (S) {$\spec{(O)} \subset \R$};
        
        \draw[thick,->,dashed] (O) -- (f) node[midway,left] {\small{$\epsilon$}};
        \draw[thick,|->,dashed] (f) -- (S) node[midway,above] {\small{$\lambda\in\Lambda$}};
        \draw[thick,|->] (OH) -- (S) node[midway,above right] {\small{$v$}};
    \end{tikzpicture}
    \caption{Given a classical embedding $\epsilon: \LHsa \ra L_\infty(\Lambda)$, every (micro)state $\lambda\in\Lambda$ defines a valuation $v_\lambda:\LHsa\ra\R$ by evaluation $v_\lambda(O) = \epsilon(O)(\lambda) \in \spec(O)$.}
    \label{fig: valuation from embedding}
\end{figure}

Given a non-degenerate spin-$1$ observable $S=\sum_{s=0,\pm 1} sp_s$ consider the functions $g_s:\R\ra\R$ defined by $g_s(x)=\delta_{sx}$. Then $g_s(v(S))=v(g_s(S))=v(p_s)$ by Eq.~(\ref{eq: FUNC}). Since $v$ must assign a unique spectral value to $S$ this implies that exactly one of the projectors $p_{-1},p_0,p_1$ is assigned the value $1$ by $v$. This argument generalises to sets of projections in any dimension: a valuation must assign every projection in such a set either the value $0$ or $1$ by the spectrum rule,\footnote{Since projections satisfy $p^2=p$, this \emph{also} follows from FUNC: $v^2(p)=v(p^2)=v(p)$ implies $v(p)\in\{0,1\}$.} and exactly one projection in every set of mutually orthogonal ones must be assigned the value $1$. Such an assignment is sometimes called a \emph{Kochen-Specker (KS) colouring}, and a set of rank-$1$ projectors in $\PH$ (equivalently rays in $\cH$) that does not admit a KS colouring is called a \emph{Kochen Specker (KS) set}.

Taking into account (the infinite set of) all rank-$1$ projections $\cP_1(\cH)$, a continuity argument based on Gleason's theorem \cite{Gleason1975} shows that $\cP_1(\cH)$ is a KS set. This was first demonstrated by Bell \cite{Bell1966} for finite-dimensional quantum systems of dimension at least three.\footnote{For a generalisation to infinite-dimensional quantum systems, see Ref.~\cite{Doering2004}.} However, it is not necessary to consider the full set of rank-$1$ projections. Ref.~\cite{KochenSpecker1967} construct a KS set comprised of 117 rays in $\R^3$.

\begin{lemma}[Bell \cite{Bell1966}, Kochen-Specker \cite{KochenSpecker1967}]\label{lm: BKS lemma}
    There exist subsets of spin-$1$ observables $\cS\subset\LHsa$ for $\dim(\cH)=3$ which admit no valuation $v:\cS\ra\R$ satisfying
    \begin{itemize}
        \item [(i)] $v(S) \in \mathrm{sp}(S)$ \hspace{1.4cm} \emph{(``spectrum rule")}
        \item [(ii)] $v(g(S)) = g(v(S))$ \hspace{0.5cm} \emph{(``functional composition (FUNC) principle" \cite{stanford})}
    \end{itemize}
    for all observables $S\in\cS$ and functions $g:\R\ra\R$.
\end{lemma}

As a consequence, we conclude that quantum theory is Kochen-Specker contextual.

\begin{theorem}[Bell \cite{Bell1966}, Kochen-Specker \cite{KochenSpecker1967}]\label{thm: BKS theorem}
    Let $d=\dim(\cH)\geq 3$. Then there exists no (classical) embedding $\epsilon:\LHsa\ra L_\infty(\Lambda)$ such that Eq.~(\ref{eq: KSNC}) holds.
\end{theorem}

However, the existence of finite Kochen-Specker sets allows for a stronger statement.

\begin{theorem}[Kochen-Specker \cite{KochenSpecker1967}]\label{thm: KS theorem}
    There exist finite subsets of spin-$1$ observables $\cS\subset\LHsa$ that admit no (classical) embedding $\epsilon:\cS\ra L_\infty(\Lambda)$ such that Eq.~(\ref{eq: KSNC}) holds.
\end{theorem}

Thm.~\ref{thm: KS theorem} gives rise to the problem of characterising Kochen-Specker contextuality.
\begin{align}\label{KS problem}
    \mathrm{\mathbf{KS\ problem:}\ characterise\ subsets\ } \cS\subset\LHsa \mathrm{\ that\ are\ KS\ contextual.} \tag{$*$}
\end{align}
Clearly, every Kochen-Specker set yields a proof of Thm.~\ref{thm: KS theorem}, which we call the ``Kochen-Specker theorem".\footnote{The main result of Ref.~\cite{KochenSpecker1967} is not stated concisely as a theorem, thus leaving room for interpretation to what one means by ``the Kochen-Specker theorem". Our interpretation differs from that in Refs.~\cite{IshamButterfieldI,DoeringFrembs2019a,BudroniEtAl2022} which call Lm.~\ref{lm: BKS lemma} the ``(Bell-)Kochen-Specker theorem". Instead, we take the main assertion of Ref.~\cite{KochenSpecker1967} to be about the impossiblity of the classical embedding problem (\ref{KS problem}).} Indeed, many more proofs of Thm.~\ref{thm: KS theorem} exist that construct all sorts of KS sets including Ref.~\cite{Mermin1990,Peres1991,ZimbaPenrose1993,Kernaghan1994,KernaghanPeres1995,CabelloEstebanzGarcia-Alcaine1997}. In particular, the size of KS sets has been drastically improved: the current record stands at 31 rays in dimension three \cite{ConwayKochen2006} and at 18 rays in dimension four \cite{CabelloEstebanzGarcia-Alcaine1997}, where the latter is also known to be optimal \cite{Zhen-PengChenGuehne2020}. This suggests that the nonexistence of valuations characterises problem (\ref{KS problem}). However, this is not the case: the search for experimental tests of contextuality in the marginal approach \cite{Pitowsky1991,BadziagBengtssonCabelloPitowsky2009,KleinmannEtAl2012,BudroniEtAl2022} has revealed an altogether different proof strategy.

\subsection{(State-independent) violation of noncontextuality inequalities}\label{sec: noncontextuality inequalities}

The ``modern view'' on Kochen-Specker (KS) contextuality studies correlations. Various frameworks with slightly different notation and conventions exist. We defer a thorough comparison between the algebraic and the marginal approach to contextuality to a companion paper \cite{Frembs2024b}. Here, we restrict our focus to resolving an apparent puzzle, first noted in Ref.~\cite{YuOh2012}, between the constraints derived from Eq.~(\ref{eq: KSNC}) for valuations, that is, Kochen-Specker sets, and proofs of the KS theorem that do not constitute KS sets.\\

\textbf{No-disturbance.} In our notation\footnote{We follow the notation in the topos approach \cite{IshamButterfieldI,IshamButterfieldII,IshamButterfieldIII,IshamButterfieldIV,IshamDoeringI,IshamDoeringII,IshamDoeringIII,IshamDoeringIV,DoeringFrembs2019a,FrembsDoering2022a,FrembsDoering2023}, which overlaps in large part with the sheaf-theoretic treatment \cite{AbramskyBrandenburger2011} and the notation in Ref.~\cite{BudroniEtAl2022}.}, a \emph{state} (also \emph{correlation}, \emph{behaviour} or \emph{empirical model} \cite{AbramskyBrandenburger2011}) of $\LHsa$ is a collection of probability distributions $\gamma=(\gamma_C)_{C\in\CH}$ with $\gamma_C:\PC\ra [0,1]$ for every context $C\in\CH$. For a state to preserve the noncontextuality constraints in Eq.~(\ref{eq: KSNC}) is to preserve the order structure in $\CH$, that is,
\begin{align}\label{eq: no-disturbance}
    \gamma_C|_{\tC}=\gamma_{\tC} = \gamma_{C'}|_{\tC} \quad\quad \forall \tC \subset C,C' \in \CH\; ,
\end{align}
where $|_\tC$ means marginalisation and is defined as $\gamma_\tC(\tp) = \gamma_C(\tp)$ for all $\tp\in\cP(\tC)\subset\PC$. Eq.~(\ref{eq: no-disturbance}) is referred to as \emph{no-disturbance} or the \emph{Gleason property} \cite{RamanathanEtAl2012,DoeringFrembs2019a,FrembsDoering2022a,Frembs2022a,FrembsDoering2023}.\\

\textbf{Marginal problem.} By Gleason's theorem \cite{Gleason1975,Doering2008,DoeringFrembs2019a}, every state, that is, every non-disturbing behaviour $\gamma=(\gamma_C)_{C\in\CH}$ of a system described by a Hilbert space $\cH$ in $\dim(\cH)\geq 3$ corresponds with a unique quantum state, that is, there exists a density matrix $\rho\in\SH$ such that $\gamma_C(p)=\gamma^\rho_C(p)=\tr[\rho p]$ for all $p\in\PC$, $C\in\CH$.\footnote{A generalisation of Gleason's theorem to the bipartite case has recently been obtained in Ref.~\cite{FrembsDoering2023}.} Clearly then, no-disturbance is not yet sufficient to restrict to classical behaviours, in the form of a correlation on $\Lambda$. For this to be the case the probability distributions $\gamma_C$ for all contexts $C\in\CH$ must arise as marginals from a single probability distribution $\mu_\gamma\in L_1(\Lambda)$ (on the state space $\Lambda$, under the embedding $\epsilon:\LHsa\ra L_\infty(\Lambda)$), that is,
\begin{align}\label{eq: marginal problem}
    \gamma_C(p)
    = \int_\Lambda d\lambda\ \mu_\gamma(\lambda) \epsilon(p)(\lambda)
    = \int_\Lambda d\lambda\ \mu_\gamma(\lambda) \chi(\lambda\mid p)
    = \mu(\Lambda_p) \; ,
\end{align}
for all $p\in\PC$ and $C\in\CH$ and where $\varepsilon(p) = \chi(\cdot\mid p)$ denotes the indicator function of the measurable subset $\Lambda_p\subset\Lambda$ in the decomposition $\Lambda = \bigcupdot_{p\in\cP_1(C)} \Lambda_p$.\footnote{Formally, it is sufficient that $\epsilon$ induces a disjoint decomposition of $\Lambda$ for any resolution of the identity $\sum_{p\in\cP_1(C)} p = \mathbbm{1}$ up to negligible sets, that is, up to sets of measure zero.} Here, $\lambda$ is thought of as a hidden variable, similar to the hidden variables describing a (classical) common cause in Bell's theorem \cite{Bell1966}. Indeed, factorisability of a behaviour is a special case of Eq.~(\ref{eq: marginal problem}), and has motivated the modern view on contextuality \cite{BudroniEtAl2022}. Both are instantiations of a \emph{marginal problem}, which has already been studied by Boole \cite{Pitowsky1989}, and more recently and systematically in terms of correlation polytopes in Ref.~\cite{Pitowsky1991,KleinmannEtAl2012}.\\

\textbf{Noncontextuality inequalities and SI-C sets.} Any correlation polytope can be described in terms of the inequalities defining its facets. In the case of the (Bell) local polytope, these inequalities are called \emph{Bell inequalities} \cite{Bell1964,Pitowsky1989,BrunnerEtAl2014}; in the more general case of the noncontextuality polytope, they are called \emph{noncontextuality inequalities} \cite{Pitowsky1991,ChavesFritz2012,KleinmannEtAl2012}. The marginal approach to contextuality is concerned with the study of the noncontextuality polytope in terms of its facet-defining inequalities. Quantum states generally violate these inequalities, thus proving the existence of non-classical correlations \cite{KlyachkoEtAl2008}. What is more, there exist noncontextuality inequalities that are violated by every quantum state \cite{YuOh2012,BengtssonBlanchfieldCabello2012,LeiferDuarte2020}. For this reason, the observables involved in such inequalities (equivalently, the generating rank-$1$ projections in their spectral decompositions) are called \emph{state-independent contextuality (SI-C) sets}. For more details, the history and present status of this approach to contextuality, we refer to the excellent review in Ref.~\cite{BudroniEtAl2022} (and references therein), as well as to the companion paper \cite{Frembs2024b}.

\subsection{Algebraic vs marginal approach}\label{sec: fundamental discrepancy}

The original approach by Kochen and Specker (in Sec.~\ref{sec: KS contextuality}) is algebraic, and motivated by logical considerations \cite{Specker1960}. In contrast, the more modern, marginal approach (in Sec.~\ref{sec: noncontextuality inequalities}) emphasises correlations over logic, and is motivated by Bell inequalities and experimental tests of contextuality (see Ref.~\cite{BartosikEtAl2009,KirchmairEtAl2009,AmselemEtAl2009,GuehneEtAl2010,MoussaEtAl2010,LapkiewiczEtAl2011} for instance). Yet, despite this difference in approach, one would expect, and consistency demands, that both notions ultimately capture equivalent views on the same underlying concept. A partial consistency check to this end can be seen in the fact that every Kochen-Specker (KS) set can be turned into a noncontextuality inequality that is violated by every quantum state \cite{BadziagBengtssonCabelloPitowsky2009}. Conversely, however, it came as a surprise when Ref.~\cite{YuOh2012} discovered an arrangement of only 13 rays, which constitute a SI-C set, that is, for which there exists an associated noncontextuality inequality that is violated by every quantum state, yet which does not constitute a KS set, that is, it does admit valuations, equivalently KS colourings. Indeed, many more SI-C sets have since been found \cite{BengtssonBlanchfieldCabello2012,XuChenSu2015,LeiferDuarte2020}.

The tension this creates has been felt immediately. For instance, Ref.~\cite{YuOh2012} argues that the algebraic structure imposed by Eq.~(\ref{eq: KSNC}) (in particular, in the form of the FUNC principle for valuations in Eq.~(\ref{eq: FUNC})) is ``too strong and unnecessary". However, if this was true it would challenge the very notion and conceptual foundation of Kochen-Specker contextuality, which is based on Eq.~(\ref{eq: KSNC}). Indeed, since Eq.~(\ref{eq: KSNC}) is fully encoded in the partial order of contexts $\CH$, which implies the no-disturbance constraints on correlations via Eq.~(\ref{eq: no-disturbance}), a modification to the notion of KS noncontextuality would necessarily also affect the very definition of the (KS) noncontextuality polytope. Viewed this way, the issue is a severe one, suggesting that we do not have a coherent understanding of this core principle of quantum mechanics.

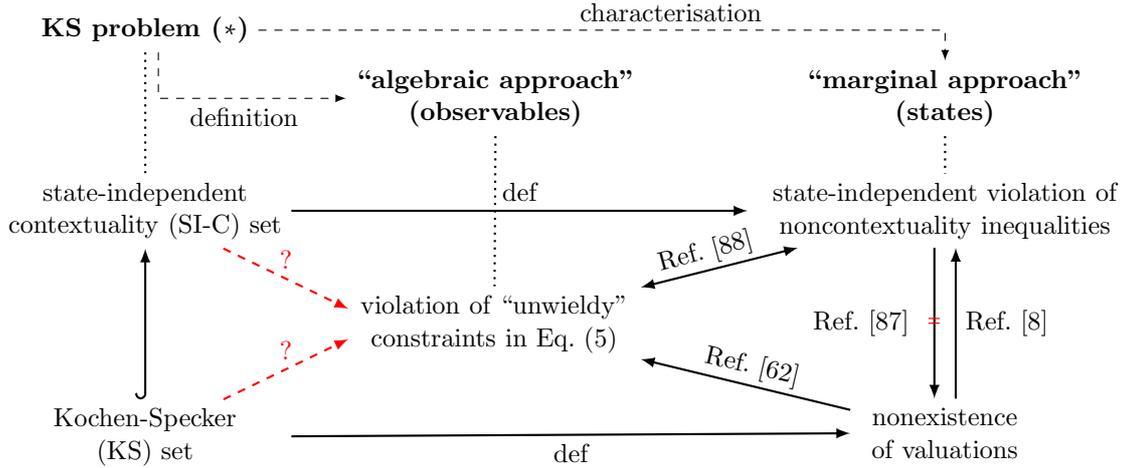
\begin{figure}[!htb]
    \centering
    \begin{tikzpicture}[every node/.style={scale=.9}]
        \node(SICS) [text width=4cm,align=center] {state-independent contextuality (SI-C) set};
        \node(KSS) [below=2cm of SICS,text width=4cm,align=center] {Kochen-Specker (KS) set};
        \node(KSNC) [below right=0.5cm and 0.75cm of SICS,text width=4cm,align=center] {violation of  ``unwieldy" constraints in Eq.~(\ref{eq: KSNC})};
        \node(SIC) [right=6cm of SICS,text width=5.5cm,align=center] {state-independent violation of
        noncontextuality inequalities};
        \node(NV) [below=2cm of SIC,text width=2.5cm,align=center] {nonexistence of valuations};
        \node(P) [above=1.6cm of SICS,align=center] {\textbf{KS problem (\ref{KS problem})}};
        \node(MA) [above=0.5cm of SIC,align=center] {\textbf{``marginal approach''}\\ \textbf{(states)}};
        \node(AA) [left=2cm of MA,align=center] {\textbf{``algebraic approach''}\\ \textbf{(observables)}};

        \draw [>=latex,<->,thick] (SIC) -- (KSNC)
        node[above=-0.25cm,midway,xshift=-0.2cm] {\rotatebox{15}{Ref.~\cite{YuGuoTong2015}}};
        \draw [-latex,thick,dashed,red] (KSS)--([yshift=-0.2cm]KSNC.west)
        node[above,midway] {?};
        \draw [-latex,thick,dashed,red] (SICS)--([yshift=0.2cm]KSNC.west)
        node[above,midway] {?};
        \draw [-latex,thick,] ([xshift=4pt]NV.north)--([xshift=4pt]SIC.south)
        node[right,midway] {Ref.~\cite{BadziagBengtssonCabelloPitowsky2009}};
        \draw [-latex,thick] ([xshift=-4pt]SIC.south)--([xshift=-4pt]NV.north)
        node[midway] {\color{red}{$=$}}
        node[left=0.2cm,midway] {Ref.~\cite{YuOh2012}};
        \draw [>=latex,->,thick] (NV)--(KSNC)
        node[above=-0.2cm,xshift=0.1cm,midway,align=center] {\rotatebox{-13}{Ref.~\cite{KochenSpecker1967}}};
        \draw [>=latex,->,thick] (SICS)--(SIC) node[midway,above] {def};
        \draw [>=latex,->,thick] (KSS)--(NV) node[midway,below] {def};
        \draw [right hook-latex,thick] (KSS)--(SICS);

        \draw [>=latex,->,dashed] ([xshift=15pt]P)|-(AA.west) node[midway,below,xshift=1.25cm] {definition};
        \draw [>=latex,->,dashed] (P)-|(MA.north) node[midway,above,xshift=-4cm] {characterisation};

        \draw [dotted,thick] (P.south)--(SICS.north);
        \draw [dotted,thick] (AA.south)--(KSNC.north);
        \draw [dotted,thick] (MA.south)--(SIC.north);
    \end{tikzpicture}
    \caption{Schematic depicting the relationship between two different perspectives on Kochen-Specker contextuality: the problem in (\ref{KS problem}) regards the characterisation of subsets of observables to which the Kochen-Specker theorem (Thm.~\ref{thm: KS theorem}) applies, and is phrased in algebraic terms, based on the defining Eq.~(\ref{eq: KSNC}) as first proposed by Kochen and Specker \cite{KochenSpecker1967}. Yet, its existing (partial) characterisations in terms of SI-C (KS) sets are obtained via the marginal approach, that is, in terms of statements about correlations and their violation of noncontextuality inequalities \cite{Pitowsky1991,SimonBruknerZeilinger2001,Larsson2002,Cabello2008b,KlyachkoEtAl2008,KleinmannEtAl2012,CabelloKleinmannBudroni2015}.}
    \label{fig: discrepancy}
\end{figure}

Of course, a simple way out of this conundrum is to accept that the existence of Kochen-Specker sets is sufficient but not necessary to prove the Kochen-Specker theorem (Thm.~\ref{thm: KS theorem}). However, this merely avoids the question why valuations do not fully characterise the Kochen-Specker problem (\ref{KS problem}), at least not to the extent of (state-independent) violations of noncontextuality inequalities. Put differently, it leaves open the problem of identifying the precise mathematical structure underlying Kochen-Specker contextuality.

At the heart of the issue seems to be the following mismatch: while the KS problem (\ref{KS problem}) is intrinsically algebraic, known (partial) characterisations of it (KS and SI-C sets) are stated in terms of states, no characterisation in terms of the observables themselves, that is, no algebraic reformulation of the ``unwieldy" constraints in Eq.~(\ref{eq: KSNC}) is available (see Fig.~\ref{fig: discrepancy}). Here, (and in more detail in Ref.~\cite{Frembs2024b}), we will resolve this mismatch. Specifically, we will provide  a complete reformulation of Eq.~(\ref{eq: KSNC}). This requires new tools, which we introduce in the next sections.

\section{Reformulating Kochen-Specker contextuality}\label{sec: wielding the unwieldy}

\subsection{The Kochen-Specker theorem revisited}\label{sec: BKS theorem revisited}

As outlined in Sec.~\ref{sec: BKS theorem}, proofs of the Kochen-Specker (KS) theorem in terms of Kochen-Specker (KS) sets reduce it to a colouring problem. We recall the details of this  argument, which upon closer inspection reveal a stronger constraint.

Kochen and Specker consider a spin-1 system and observe that triples of squared spin-$1$ observables $S^2_x,S^2_y,S^2_z$, corresponding to measurements in orthogonal directions in $\R^3$, share the same spectrum $\spec(S^2_x) = \spec(S^2_y) = \spec(S^2_z) = \{0,1\}$, mutually commute $[S^2_x,S^2_y] =[S^2_y,S^2_z] = [S^2_z,S^2_x] = 0$, and satisfy the algebraic constraint $S^2=S^2_x + S^2_y + S^2_z = 2\mathbbm{1}$. It follows that a valuation must assign two squared spin-1 observables in every triple the value $1$ and the remaining one the value $0$. Under the (implicitly assumed) identification of (squared) spin-1 observables with directions in space, these properties thus define a \emph{Kochen-Specker (KS) colouring} of the unit sphere, that is a map $\R^3 \supset S^2 \ra \{0,1\}$ such that (i) in every set of mutually orthogonal rays one is assigned one of two colours, say red (corresponding to the value $0$ above), and the other two vectors are assigned another colour, say green (corresponding to the value $1$ above).

By Lm.~\ref{lm: BKS lemma}, no KS colouring under the above identification of (squared) spin-1 observables with directions in $\R^3$ exists. It is therefore impossible to consistently label one of the (squared) spin-1 observables in every commuting triple by e.g. `$S_z$'. This is the price we must pay in insisting on a classical embedding: it is necessarily contextual in its labelling of these observables. However, from this perspective a KS colouring is clearly a restriction: assuming a classical embedding, we ought to be able to \emph{label all squared (and thus all) spin-1 observables consistently}. In the next sections, we develop this intuition.

\subsection{Context connections and context cycles}\label{sec: context connection}

In this section, we formalise the idea sketched in the previous section on labelling (squared) spin-1 observables. Since an observable is uniquely determined by its spectral decomposition, we will do so by making explicit the identifications between projections across different (maximal) contexts.

\begin{definition}\label{def: context connection}
    Let $\CSpin\subset\CH$ be the partial order of contexts of a quantum subsystem, generated by $\cS\subset\LHsa$. A \emph{(context) connection $\fl = (l_{C'C})_{C,C'\in\CSm}$ on $\CSpin$} is a collection of bijective maps $l_{C'C}: \mc{P}_1(C) \ra \mc{P}_1(C')$ with $l_{CC'} = l^{-1}_{C'C}$ and such that
    \begin{align}\label{eq: noncontextual context connection}
        l_{C'C}|_{\mc{P}_1(C\cap C')} = \id \quad \forall C,C' \in \CSm\; .
    \end{align}
\end{definition}

By fixing the elements $l_{C'C}$ of a context connection $\fl$ on common projections $\mc{P}_1(C\cap C')$ in subcontexts $C \cap C'$, Eq.~(\ref{eq: noncontextual context connection}) expresses a consistency condition for $\fl$ with the order structure of $\CSpin$. For a schematic representation of a context connection, see Fig.~\ref{fig: context connection}.

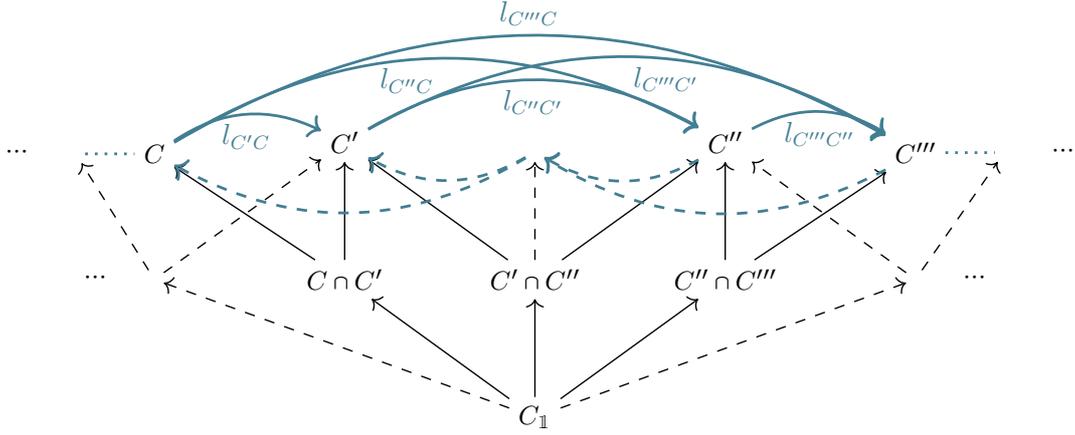
\begin{figure}
    \centering
    \scalebox{1.3}{\begin{tikzpicture}[node distance=2.75cm, every node/.style={scale=0.7}]
    \node(Vp1)                        {};
    \node(Va2c)   [above= 1cm of Vp1]        {$C$};
    \node(V1bc)   [left=0.5cm of Va2c]        {};
    \node(Vp2)   [right of=Vp1]        {$C\cap C'$};
    \node(Vp3)   [right of=Vp2]        {$C' \cap C''$};
    \node(Vab3)   [above= 1cm of Vp3]        {};
    \node(V123)   [above= 1cm of Vp2]        {$C'$};
    
    \node(Vp'2)   [right of=Vp3]      {$C'' \cap C'''$};
    \node(Vp'1)   [right of=Vp'2]       {};
    \node(Va2'c)   [above=1cm of Vp'1]        {$C'''$};
    \node(V1'bc)   [right=0.5cm of Va2'c]        {};
    \node(V1'2'3)   [above= 1cm of Vp'2]      {$C''$};
    
    \node(Vone)  [below= 1cm of Vp3]  {$C_\one$};
    
    \node(3DotsUpLeft)     [left=0.3cm of V1bc]    {$\cdots$};
    \node(3DotsUpRight)     [right=0.3cm of V1'bc]    {$\cdots$};
    \node(3DotsDownLeft)   [left=0.3cm of Vp1]    {$\cdots$};
    \node(3DotsDownRight)   [right=0.3cm of Vp'1]    {$\cdots$};

    \draw [->,dashed](Vone) -- (Vp1);
    \draw [->](Vone) -- (Vp2);
    \draw [->](Vone) -- (Vp3);
    \draw [->](Vone) -- (Vp'2);
    \draw [->,dashed](Vone) -- (Vp'1);
    
    \draw [->,dashed](Vp1) -- (V123);
    \draw [->,dashed](Vp1) -- (V1bc);
    \draw [->](Vp2) -- (V123);
    \draw [->](Vp2) -- (Va2c);

    \draw [->](Vp3) -- (V123);
    \draw [->](Vp3) -- (V1'2'3);
    \draw [->,dashed](Vp3) -- (Vab3);
    
    \draw [->](Vp'2) -- (V1'2'3);
    \draw [->](Vp'2) -- (Va2'c);
    \draw [->,dashed](Vp'1) -- (V1'2'3);
    \draw [->,dashed](Vp'1) -- (V1'bc);

    \draw [->,db,thick] (Va2c) to[bend left] node[midway,below] {\large\textbf{$l_{C'C}$}} (V123);
    \draw [->,db,thick] (V123) to[bend left] node[midway,below] {\large\textbf{$l_{C''C'}$}} (V1'2'3);
    \draw [->,db,thick] (Va2c) to[bend left] node[midway,below] {\large\textbf{$l_{C''C}\ \ \ \ \ $}} (V1'2'3);
    \draw [->,db,thick] (V1'2'3) to[bend left] node[midway,below] {\large\textbf{$l_{C'''C''}$}} (Va2'c);
    \draw [->,db,thick] (Va2c) to[bend left] node[midway,above] {\large\textbf{$l_{C'''C}$}} (Va2'c);
    \draw [->,db,thick,dashed] (Va2'c) to[bend left] (Vab3);
    \draw [->,db,thick] (V123) to[bend left] node[midway,below] {\large\textbf{$\ \ \ \ \ \ \ l_{C'''C'}$}} (Va2'c);

    \draw [->,db,thick,dashed] (Vab3) to[bend left] (V123) ;
    \draw [->,db,thick,dashed] (V1'2'3) to[bend left] (Vab3) ;
    \draw [->,db,thick,dashed] (Vab3) to[bend left] (Va2c);

    \draw [-,db,thick,dotted] (V1bc) -- (Va2c) ;
    \draw [-,db,thick,dotted] (Va2'c) -- (V1'bc) ;
\end{tikzpicture}}
    \caption{Schematic of a context connection $\fl = (l_{C'C})_{C,C'\in\CSm}$ (blue) on a context category $\CSpin\subset\cC(\C^3)$ generated by a set of spin-$1$ observables $\cS\subset\LHsa$ (black). $\fl$ preserves subcontexts (generated by common squared spin-$1$ observables), that is, $l_{C'C}|_{\mc{P}_1(C\cap C')} = \mathrm{id}$ for all maximal contexts $C,C'\in\CSm$.}
    \label{fig: context connection}
\end{figure}

Next, in order to express the constraints imposed by Eq.~(\ref{eq: KSNC}) in terms context connections on $\CSpin\subset\CH$, we consider their action along context cycles (see Fig.~\ref{fig: context cycle}).

\begin{definition}\label{def: context cycle}
    A \emph{context cycle $\gamma$ in $\CSm\subset\CHm$} is a tuple $\gamma=(C_0,\cdots,C_{n-1})$ of 
    maximal contexts $C_i\in\CSm$, together with subcontexts $C_{i\cap i+1}=C_i \cap C_{i+1}$ for $i\in\zz_n$.\footnote{Here, $\zz_n$ denotes the abelian group under addition modulo $n$, in particular, $(n-1)+1 = 0 \in \zz_n$.}
\end{definition}

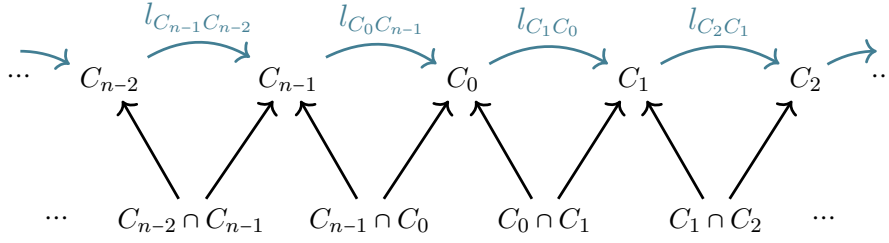
\begin{figure}
    \centering
    \scalebox{1.3}{\begin{tikzpicture}[node distance=3cm, every node/.style={scale=0.75}]
    \node(V-2)   {$C_{n-2}$};
    \node(V-1)   [right=1cm of V-2]        {$C_{n-1}$};
    \node(V0)   [right=1.1cm of V-1]        {$C_0$};
    \node(V1)   [right=1.2cm of V0]        {$C_1$};
    \node(V2)   [right=1.2cm of V1]        {$C_2$};
    
    \node(V-21)   [below=1cm of V-2, xshift=1.1cm]   {$C_{n-2}\cap C_{n-1}$};
    \node(V-10)   [below=1cm of V-1, xshift=1.1cm]   {$C_{n-1}\cap C_0$};
    \node(V01)   [below=1cm of V0, xshift=1.1cm]   {$C_0\cap C_1$};
    \node(V12)   [below=1cm of V1, xshift=1.1cm]   {$C_1\cap C_2$};
    
    \node(3DotsUpLeft)     [left=0.3cm of V-2]    {$\cdots$};
    \node(3DotsUpRight)     [right=0.3cm of V2]    {$\cdots$};
    \node(3DotsDownLeft)   [left=0.3cm of V-21]    {$\cdots$};
    \node(3DotsDownRight)   [right=0.3cm of V12]    {$\cdots$};

    \draw [->,thick] (V-21) -- (V-2);
    \draw [->,thick] (V-21) -- (V-1);
    \draw [->,thick] (V-10) -- (V-1);
    \draw [->,thick] (V-10) -- (V0);
    \draw [->,thick] (V01) -- (V0);
    \draw [->,thick] (V01) -- (V1);
    \draw [->,thick] (V12) -- (V1);
    \draw [->,thick] (V12) -- (V2);

    \draw [->,db,thick] (V-2) to[bend left] node[midway,above] {\large$l_{C_{n-1}C_{n-2}}$} (V-1);
    \draw [->,db,thick] (V-1) to[bend left] node[midway,above] {\large$l_{C_0C_{n-1}}$} (V0);
    \draw [->,db,thick] (V0) to[bend left] node[midway,above] {\large$l_{C_1C_0}$} (V1);
    \draw [->,db,thick] (V1) to[bend left] node[midway,above] {\large$l_{C_2C_1}$} (V2);

    \node(V-2x)   [left=-0.1cm of V-2,yshift=0.2cm]   {};
    \draw [<-,db,thick] (V-2x) arc (60:90:1);
    \node(V2x)   [right=-0.15cm of V2,yshift=0.2cm]   {};
    \draw [->,db,thick] (V2x) arc (120:90:1);
\end{tikzpicture}}
    \caption{Schematic of a context cycle $\gamma=(C_0,\cdots,C_{n-1})$ with $C_i\in\CHm$ and $C_{i \cap (i+1)} = C_{i+1}\cap C_i$ for all $i\in\zz_n$ (black), and elements of a context connection $\fl$ (blue).}
    \label{fig: context cycle}
\end{figure}

\subsection{Casting Kochen-Specker contextuality in geometric form}\label{sec: CNC}

Kochen and Specker deem the constraints in Eq.~(\ref{eq: KSNC}) ``too unwieldy''. To date, no alternative algebraic formulation of Kochen-Specker noncontextuality has emerged. And at the absence of a unified mathematical framework of Kochen-Specker contextuality, the original algebraic approach has thus been largely abandoned, in favour of the study of contextual correlations \cite{BudroniEtAl2022} (see also Fig.~\ref{fig: discrepancy}). In contrast to this development, we now show that the concept of context connections (see Sec.~\ref{sec: context connection}) allows to re-express Kochen-Specker noncontextuality in a new, and inherently geometric form.

For simplicity and conceptual clarity, here we focus on the case of spin-$1$ observables.

\begin{theorem}\label{thm: CNC}
    Let $\cS\subset\LHsa$ be a subset of spin-$1$ observables in $\dim(\cH)=3$, and denote by $\CSpin\subset\CH$ the partial order of contexts generated by $\cS$. If $\cS$ admits a (classical) embedding $\epsilon:\cS\ra L_\infty(\Lambda)$ satisfying Eq.~(\ref{eq: KSNC}), then there exists a context connection $\fl = (l_{C'C})_{C,C'\in\CSm}$ on $\CSpin$ which satisfies the \emph{triviality constraints},
    \begin{equation}\label{eq: CNC}
        \circ_{i=0}^{n-1} l_{C_{(i+1)}C_i}=\id\; ,
    \end{equation}
    for every context cycle $\gamma = (C_0,\cdots,C_{n-1})$ in $\CSm$.
\end{theorem}

\begin{proof}
    $\epsilon$ maps the projections $p\in\PS\subset\PH$ (that arise in spectral resolutions of the spin-$1$ observables in $\cS$) to measurable subsets $\Lambda_p\subset\Lambda$ such that (up to negligible sets) $\Lambda_0=\emptyset$, $\Lambda_\mathbbm{1}=\Lambda$ and $\Lambda_{p+p'}=\Lambda_p \cupdot \Lambda_{p'}$ whenever $pp'=0$. For every spin-$1$ observable $S\in\cS$ with spectral decomposition $S=\sum_{s=0,\pm 1} sp_s$, $p_s\in\cP_1(\cS)$, $\epsilon$ therefore induces a partition $\Lambda = \bigcupdot_{s=0,\pm 1} \Lambda^S_s$ into measurable subsets $\Lambda^S_s$, which yield the outcome $s\in \spec(S)$ under evaluation of $f_S:=\epsilon(S):\Lambda\ra\R$, where $\im(f_S)=\spec(S)=\{-1,0,+1\}$.\footnote{Formally, $\epsilon|_\PS:\PS\ra\cB(\Lambda)$ defines an orthomorphism into the Boolean $\sigma$-algebra $\cB(\Lambda)$.}
    
    Consider two spin-$1$ observables $S,S'\in\cS$, and let $\phi_{S'S}:\Lambda\ra\Lambda$ be a measurable function such that $\phi_{S'S}(\Lambda^S_s)\subset\Lambda^{S'}_{s'=s}$ for all $s\in\{-1,0,+1\}$. Clearly, this definition depends on the choice of generating spin-$1$ observables for the contexts $C=C(S),C'=C(S')$. In other words, it depends on a choice of bijection $l_{C'C}: \mc{P}_1(C) \ra \mc{P}_1(C')$. We may thus write $\phi_{S'S} = \phi^{l_{C'C}}_{C'C}$, and $\phi^\fl=(\phi^{l_{C'C}}_{C'C})_{C,C'\in\CSm}$ for a collection depending on the context connection $\fl=(l_{C'C})_{C,C'\in\CSm}$. We will derive constraints on context connections $\fl$ by concatenating the maps in $\phi^\fl$ along context cycles $(C_0,\cdots,C_{n-1})$, in which case we will abbreviate the notation to $\phi^{l_{(i+1)i}}_{(i+1)i}:=\phi^{l_{C_{(i+1)}C_i}}_{C_{(i+1)}C_i}=\phi_{S_{(i+1)}S_i}$ for $C_i=C(S_i)$ and $i\in\zz_n$.
    
    To this end, note first that $\phi_{S_{(i+1)}S_i}$ can generally not be chosen to be one-to-one, since $\Lambda^{S_i}_{s_i=s}$ and $\Lambda^{S_{(i+1)}}_{s_{i+1}=s}$ will generally have different cardinalities. Nevertheless, by restricting to measurable subsets $\tilde{\Lambda}^{S_i}_{s_i} \subset \Lambda^{S_i}_{s_i}$ such that there exist injections $\tilde{\Lambda}^{S_i}_{s_i} \hookrightarrow \Lambda^{S_j}_{s_j=s_i}$ for all $s_i\in\{-1,0,+1\}$ and $i,j\in\zz_n$, we can find invertible (measurable) maps $\phi_{S_{(i+1)}S_i}: \Lambda\ra\Lambda$ such that $\phi_{S_{(i+1)}S_i}(\tilde{\Lambda}^{S_i}_{s_i}) = \tilde{\Lambda}^{S_{(i+1)}}_{s_{(i+1)}=s_i}$ for all $s_i\in\{-1,0,+1\}$ and such that $\phi_{S_{(i+1)}S_i}|_{\Lambda/\tilde{\Lambda}^{S_i}} = \mathrm{id}$ for all $i\in\zz_n$, where $\tilde{\Lambda}^{S_i} = \bigcupdot_{s_i=0,\pm 1} \tilde{\Lambda}^{S_i}_{s_i}$.\footnote{Indeed, we may choose $\tilde{\Lambda}^S_s\subset\Lambda^S_s$ to be of cardinality $\min_{S'\in\cS} |\Lambda^{S'}_{s'=s}|$ for all $s\in\{-1,0,+1\}$ and $S\in\cS$.}
    
    Denote by $f_i=\epsilon(S_i)$ the measurable function on $\Lambda$ representing the spin-$1$ observable $S_i$ under the classical embedding $\epsilon$. Then $f_{(i+1)}|_{\tilde{\Lambda}^{S_{(i+1)}}} = f_i|_{\tilde{\Lambda}^{S_i}} \circ  (\phi^{l_{(i+1)i}}_{(i+1)i})^{-1}$, hence,\footnote{What is more, since $\phi^{l_{(i+1)i}}_{(i+1)i}|_{\Lambda/\tilde{\Lambda}^{S_i}}=\mathrm{id}$ we also have $f_0\circ\left(\circ_{i=0}^{n-1} \phi^{l_{(i+1)i}}_{(i+1)i}\right)^{-1} = f_0$.}
    \begin{equation*}
        f_0|_{\tilde{\Lambda}^{S_0}} \circ \left(\circ_{i=0}^{n-1} \phi^{l_{(i+1)i}}_{(i+1)i}\right)^{-1} = f_0|_{\tilde{\Lambda}^{S_0}}
        \ \ \Llra \ \ \circ_{i=0}^{n-1} \phi^{l_{(i+1)i}}_{(i+1)i} = \id
        \ \ \Llra \ \ \circ_{i=0}^{n-1} l_{(i+1)i} = \id\; ,
    \end{equation*}
    where we used that $f_0 = \epsilon(S_0)$ is non-degenerate with respect to the $\{\Lambda^{S_0}_s\}_{s=0,\pm 1}$-partition ($S_0$ is non-degenerate and $\im(f_0)=\spec(S_0)=\{-1,0,+1\})$. Since we constructed $\phi^\fl$ simply by assuming the existence of a classical embedding $\epsilon$, the result follows.
\end{proof}

In Ref.~\cite{Frembs2024b} we further generalise Thm.~\ref{thm: CNC} beyond the quantum case to include partial algebras \cite{KochenSpecker1967}, thus showing that the triviality constraints in Eq.~(\ref{eq: CNC}) encode Kochen-Specker noncontextuality completely. This will also enable us to provide a more detailed comparison with the marginal approach to contextuality, specifically to prove that Thm.~\ref{thm: CNC} fully reconciles Kochen-Specker contextuality with the notion of state-independent contextuality (SI-C) underlying SI-C sets in the marginal, and specifically in the graph-theoretic approach. In the present paper, we will make this point by showing explicitly how it rules out SI-C sets 
such as in Ref.~\cite{YuOh2012} (Cor.~\ref{cor: Yu-Oh for contexts} below).

\subsection{From KS-colourings to three-colourings}\label{sec: resolving the discrepancy}

In Sec.~\ref{sec: BKS theorem revisited}, we argued that a Kochen-Specker (KS) colouring can be understood in terms of a labelling of the squared spin-$1$ observables by a spatial direction in $\R^3$, e.g. denoting one observable in every triple of squared spin-1 observables by `$S^2_z$'. Thm.~\ref{thm: CNC} justifies this view, in fact, it demonstrates that \emph{KS noncontextuality (Eq.~(\ref{eq: KSNC})) can be understood as the possibility to label all squared spin-1 observables consistently in this way}.

\begin{theorem}\label{thm: KSNC = 3-colouring}
    Let $\cS\subset\LHsa$ be a subset of spin-$1$ observables in $\dim(\cH)=3$, and denote by $\CSpin\subset\CH$ the partial order of contexts generated by $\cS$. If $\cS$ admits a (classical) embedding $\epsilon:\cS\ra L_\infty(\Lambda)$ satisfying Eq.~(\ref{eq: KSNC}), then every squared spin-$1$ observable in $\cS$ can be labelled by either `$S^2_x$', `$S^2_y$' or `$S^2_z$' such that in every commuting triple each label appears exactly once.
\end{theorem}

\begin{proof}
    Fix a maximal context $C\in\CSm$ and a generating spin-$1$ observable $S$ with $C=C(S)$ and spectral decomposition $S=\sum_{s=0,\pm 1} sp_s$. Without loss of generality we identify its eigenspaces with spatial directions via $p_{-1}\leftrightarrow x$, $p_1\leftrightarrow y$ and $p_{0}\leftrightarrow z$. Since any context connection $\fl=(l_{C'C})_{C,C'\in\CSm}$ relates between the generating projections in maximal contexts, $\fl$ also defines a unique generating spin-$1$ observable in every other maximal context $C\neq C'=C(S')\in\CSm$ by $S'=l_{C(S'),C(S)}(S)=\sum_{s=0,\pm 1} sl_{C(S'),C(S)}(p_s)$. Moreover, let $S_\pi=\sum_{s=0,\pm 1} \pi(s)p_s$ denote the spin-$1$ observables obtained from $S$ by a permutation of its spectrum. Then $C(S)=C(S_\pi)$, and $\fl$ further defines unique elements $S'_\pi = l_{C(S'),C(S)}(S_\pi) = \sum_{s=0,\pm 1} \pi(s)l_{C(S'),C(S)}(p_s)$. We may thus label the eigenspaces of every other spin-$1$ observable in terms of those of $S$, the action of the symmetric group $\Sym(S)$ acting on its spectrum, and the context connection $\fl$.
    
    By Thm.~\ref{thm: CNC}, there exists a context connection $\fl=(l_{C'C})_{C,C'\in\CSm}$ satisfying the constraints in Eq.~(\ref{eq: CNC}). This implies that our labelling above does not depend on the reference observable $S$, that is, we also have $S''_\pi=l_{C(S''),C(S')}(S'_\pi)$ for all contexts $C\neq C',C''\in\CSm$ and generating observables defined as above. Finally, since every non-maximal context $C_\mathbbm{1}\neq C\in\CSpin$ is generated by a unique squared spin-$1$ observable, and since $\fl$ preserves subcontexts by definition, this yields the desired labelling of all squared spin-$1$ observables $S^2$ in terms of the eigenspaces of spin-$1$ observables $S\in\cS$.
\end{proof}

Notably, by identifying every squared spin-$1$ observable with the rank-$1$ projection onto its $0$-eigenspace, the labelling of squared spin-$1$ observables in Thm.~\ref{thm: KSNC = 3-colouring} can be cast as a colouring problem of these projections. This should be contrasted with the colouring problem defined for valuations in Sec.~\ref{sec: BKS theorem}: the latter defines a two-colouring, while the former defines a three-colouring.\footnote{In Ref.~\cite{Frembs2024b}, we prove that Kochen-Specker noncontextuality is equivalent to a colouring problem of the orthogonality graph of its minimal projections more generally, and use this to re-derive and generalise similar results obtained in the graph-theoretic approach to contextuality \cite{RamanathanHorodecki2014,CabelloKleinmannBudroni2015}.} This shows that the constraints in Eq.~(\ref{eq: KSNC}) are strictly stronger than the constraints on valuations in Lm.~(\ref{lm: BKS lemma}).

Thm.~\ref{thm: CNC} characterises KS noncontextuality in terms of constraints on context connections on the partial order of contexts. This is in contrast to the correlation approach in terms of noncontextuality inequalities \cite{BudroniEtAl2022}. Both capture different notions of contextuality (see Ref.~\cite{Frembs2024b} for a comparison). As a corollary to Thm.~\ref{thm: KSNC = 3-colouring}, we now show that the SI-C set in Ref.~\cite{YuOh2012} yields a proof of the KS theorem - not in terms of the violation of a noncontextuality inequality it induces, but as a direct consequence of Eq.~(\ref{eq: KSNC}).

\begin{corollary}\label{cor: Yu-Oh for contexts}
    The spin-$1$ observables $\cS_\mathrm{YO}\subset\LHsa$ with $\dim(\cH)=3$ do not admit a classical embedding $\epsilon:\cS_\mathrm{YO}\ra L_\infty(\Lambda)$ for any measurable space $\Lambda$.
\end{corollary}

\begin{proof}
    The 13 rays in Ref.~\cite{YuOh2012} define a set of rank-$1$ projections, which (together with the uniquely determined, complementary rank-$1$ projection for every pair of orthogonal projections) defines a subset of spin-$1$ observables $\cS_\mathrm{YO}\subset\LHsa$. $\cS_\mathrm{YO}$ does not admit a three-colouring \cite{Cabello2012b} (see App.~\ref{app: Yu-Oh} for an explicit proof). It thus follows from Thm.~\ref{thm: KSNC = 3-colouring} that $\cS_\mathrm{YO}$ does not admit a classical embedding, hence, is Kochen-Specker contextual.
\end{proof}

Thm.~\ref{thm: KSNC = 3-colouring} thus reconciles Kochen-Specker noncontextuality, defined in terms of the algebraic constraints on classical embeddings in Eq.~(\ref{eq: KSNC}), with modern proofs of the Kochen-Specker theorem (Thm.~\ref{thm: KS theorem}) in the form of (state-independent) violations of noncontextuality inequalities.

\section{Discussion}\label{sec: discussion}

We presented a reformulation of the (``unwieldy") algebraic constraints in Eq.~(\ref{eq: KSNC}) on classical embeddings of quantum systems, as defined in the seminal work by Kochen and Specker \cite{KochenSpecker1967}. To this end, we have introduced a conceptually and mathematically new tool, called a \emph{context connection}, which allowed us to express Kochen-Specker noncontextuality in terms of ``triviality constraints" along \emph{context cycles}. What is more, this tool offers fresh perspectives on a variety of related research directions.

First, for what regards the Kochen-Specker theorem, we have demonstrated how Thm.~\ref{thm: CNC} fills the gap between traditional proofs, based on the nonexistence of valuations for Kochen-Specker (KS) sets, and proofs based on state-independent contextuality (SI-C) sets \cite{YuOh2012,BengtssonBlanchfieldCabello2012,LeiferDuarte2020}, which arose out of the study of experimental tests of contextuality \cite{BartosikEtAl2009,KirchmairEtAl2009,AmselemEtAl2009,GuehneEtAl2010,MoussaEtAl2010,LapkiewiczEtAl2011}. Thm.~\ref{thm: CNC} thus reconciles the original algebraic approach, with the modern marginal approach to contextuality \cite{BudroniEtAl2022}. For an in-depth comparison between these two approaches, we refer to the companion paper \cite{Frembs2024b} to this work.

Second, the new concept of a context connection grounds Kochen-Specker contextuality on a solid mathematical footing, which previously had been lacking \cite{BudroniEtAl2022,JokinenEtAl2024}. Our framework thus opens the possibility towards unification with various other approaches to contextuality in the literature \cite{Spekkens2005,AbramskyBrandenburger2011,CabelloSeveriniWinter2014,DzhafarovKujalaCervantes2016,Raussendorf2016,CleveLiuSlofstra2017,OkayRoberts2017,FrembsRobertsBartlett2018,OkayRaussendorf2020}. In particular, we expect the tools developed here to significantly improve ongoing efforts to quantitatively analyse contextuality as a resource in quantum computation \cite{Raussendorf2013,HowardEtAl2014,BravyiGossetKoenig2018,FrembsRobertsCampbellBartlett2023}.

Finally, our reformulation of Eq.~(\ref{eq: KSNC}) in Thm.~\ref{thm: CNC} highlights the intrinsically geometric nature of Kochen-Specker contextuality: Eq.~(\ref{eq: CNC}) may be read as the condition that the ``space of contexts" $\CSpin$ has trivial holonomy (in a suitably generalised sense). Exploring this idea further may lead to a deeper understanding of contextuality as a ``geometric obstruction'' \cite{Specker1960}, and towards a geometric reformulation of quantum theory.

\bibliographystyle{abbrv}
\bibliography{bibliography}

\begin{thebibliography}{10}

\bibitem{AbramskyBrandenburger2011}
S.~Abramsky and A.~Brandenburger.
\newblock The sheaf-theoretic structure of non-locality and contextuality.
\newblock {\em New J. Phys.}, 13(11):113036, 2011.

\bibitem{AmselemEtAl2009}
E.~Amselem, M.~R\aa{}dmark, M.~Bourennane, and A.~Cabello.
\newblock State-independent quantum contextuality with single photons.
\newblock {\em Phys. Rev. Lett.}, 103:160405, Oct 2009.

\bibitem{IshamDoeringI}
{Anderas D{\"o}ring, and Chris J. Isham}.
\newblock A topos foundation for theories of physics: I. formal languages for
  physics.
\newblock {\em J. Math. Phys.}, 49(5):053515, 2008.

\bibitem{IshamDoeringII}
{Andreas D{\"o}ring, and Chris J. Isham}.
\newblock {A topos foundation for theories of physics: II. Daseinisation and
  the liberation of quantum theory}.
\newblock {\em J. Math. Phys.}, 49(5):053516, 2008.

\bibitem{IshamDoeringIII}
{Andreas D{\"o}ring, and Chris J. Isham}.
\newblock {A topos foundation for theories of physics: III. The representation
  of physical quantities with arrows
  $\breve{\delta}(\hat{A}):\underline{\Sigma} \rightarrow
  \underline{\mathbb{R}^\leftrightarrow}$}.
\newblock {\em J. Math. Phys.}, 49(5):053517, 2008.

\bibitem{IshamDoeringIV}
{Andreas D{\"o}ring, and Chris J. Isham}.
\newblock {A topos foundation for theories of physics: IV. Categories of
  systems}.
\newblock {\em J. Math. Phys.}, 49(5):053518, 2008.

\bibitem{Arnold2013}
V.~Arnold, K.~Vogtmann, and A.~Weinstein.
\newblock {\em Mathematical Methods of Classical Mechanics}.
\newblock Graduate Texts in Mathematics. Springer New York, 2013.

\bibitem{BadziagBengtssonCabelloPitowsky2009}
P.~Badziag, I.~Bengtsson, A.~Cabello, and I.~Pitowsky.
\newblock Universality of state-independent violation of correlation
  inequalities for noncontextual theories.
\newblock {\em Phys. Rev. Lett.}, 103:050401, Jul 2009.

\bibitem{BartosikEtAl2009}
H.~Bartosik, J.~Klepp, C.~Schmitzer, S.~Sponar, A.~Cabello, H.~Rauch, and
  Y.~Hasegawa.
\newblock Experimental test of quantum contextuality in neutron interferometry.
\newblock {\em Phys. Rev. Lett.}, 103:040403, Jul 2009.

\bibitem{Bell1964}
J.~S. Bell.
\newblock On the {E}instein-{P}odolsky-{R}osen paradox.
\newblock {\em Physics}, 1:195, Nov 1964.

\bibitem{Bell1966}
J.~S. Bell.
\newblock On the problem of hidden variables in quantum mechanics.
\newblock {\em Rev. Mod. Phys.}, 38:447--452, Jul 1966.

\bibitem{Bell1982}
J.~S. Bell.
\newblock On the impossible pilot wave.
\newblock {\em Found. Phys.}, 12(10):989--999, Oct 1982.

\bibitem{BeltramettiBugajski1995}
E.~G. Beltrametti and S.~Bugajski.
\newblock A classical extension of quantum mechanics.
\newblock {\em J. Phys. A}, 28(12):3329, Jun 1995.

\bibitem{BengtssonBlanchfieldCabello2012}
I.~Bengtsson, K.~Blanchfield, and A.~Cabello.
\newblock {A Kochen-Specker inequality from a SIC}.
\newblock {\em Physics Letters A}, 376(4):374--376, 2012.

\bibitem{Bohm1952}
D.~Bohm.
\newblock A suggested interpretation of the quantum theory in terms of
  ``hidden" variables. i.
\newblock {\em Phys. Rev.}, 85:166--179, Jan 1952.

\bibitem{BravyiGossetKoenig2018}
S.~Bravyi, D.~Gosset, and R.~K{\"o}nig.
\newblock Quantum advantage with shallow circuits.
\newblock {\em Science}, 362(6412):308--311, 2018.

\bibitem{BrunnerEtAl2014}
N.~Brunner, D.~Cavalcanti, S.~Pironio, V.~Scarani, and S.~Wehner.
\newblock Bell nonlocality.
\newblock {\em Rev. Mod. Phys.}, 86:419--478, Apr 2014.

\bibitem{BudroniEtAl2022}
C.~Budroni, A.~Cabello, O.~G\"uhne, M.~Kleinmann, and J.-A. Larsson.
\newblock {Kochen-Specker contextuality}.
\newblock {\em Rev. Mod. Phys.}, 94:045007, Dec 2022.

\bibitem{BuschLahtiPellonpaa_QuantumMeasurement}
P.~Busch, P.~Lahti, J.~Pellonp{\"a}{\"a}, and K.~Ylinen.
\newblock {\em Quantum Measurement}.
\newblock Theoretical and Mathematical Physics. Springer, 2016.

\bibitem{IshamButterfieldII}
J.~Butterfield and C.~J. Isham.
\newblock {A topos perspective on the Kochen-Specker Theorem II. Conceptual
  aspects and classical analogues}.
\newblock {\em Int. J. Theor. Phys.}, 38(3):827--859, 1999.

\bibitem{IshamButterfieldIV}
J.~Butterfield and C.~J. Isham.
\newblock {Topos perspective on the Kochen-Specker theorem: IV. Interval
  valuations}.
\newblock {\em Int. J. Theor. Phys.}, 41(4):613--639, 2002.

\bibitem{Cabello2008b}
A.~Cabello.
\newblock Experimentally testable state-independent quantum contextuality.
\newblock {\em Phys. Rev. Lett.}, 101:210401, Nov 2008.

\bibitem{Cabello2012b}
A.~Cabello.
\newblock State-independent quantum contextuality and maximum nonlocality,
  2012.

\bibitem{CabelloEstebanzGarcia-Alcaine1997}
A.~Cabello, J.~M. Estebaranz, and G.~Garc{\'i}a-Alcaine.
\newblock {Bell-Kochen-Specker theorem: A proof with 18 vectors}.
\newblock {\em Phys. Lett., A}, 212(4):183--187, 1996.

\bibitem{CabelloKleinmannBudroni2015}
A.~Cabello, M.~Kleinmann, and C.~Budroni.
\newblock Necessary and sufficient condition for quantum state-independent
  contextuality.
\newblock {\em Phys. Rev. Lett.}, 114:250402, Jun 2015.

\bibitem{CabelloSeveriniWinter2014}
A.~Cabello, S.~Severini, and A.~Winter.
\newblock Graph-theoretic approach to quantum correlations.
\newblock {\em Phys. Rev. Lett.}, 112:040401, Jan 2014.

\bibitem{ChavesFritz2012}
R.~Chaves and T.~Fritz.
\newblock Entropic approach to local realism and noncontextuality.
\newblock {\em Phys. Rev. A}, 85:032113, Mar 2012.

\bibitem{ChiribellaYuan2014}
G.~Chiribella and X.~Yuan.
\newblock Measurement sharpness cuts nonlocality and contextuality in every
  physical theory, 2014.

\bibitem{HLSBohrification2009}
N.~P.~L. Chris~Heunen and B.~Spitters.
\newblock Bohrification.
\newblock {\em Deep Beauty}, 2009.

\bibitem{CleveLiuSlofstra2017}
R.~Cleve, L.~Liu, and W.~Slofstra.
\newblock Perfect commuting-operator strategies for linear system games.
\newblock {\em J. Math. Phys.}, 58(1):012202, 2017.

\bibitem{ConwayKochen2006}
J.~Conway and S.~Kochen.
\newblock The free will theorem.
\newblock {\em Found. Phys.}, 36(10):1441--1473, Oct 2006.

\bibitem{CrullBacciagaluppi2016}
E.~Crull and G.~Bacciagaluppi, editors.
\newblock {\em Grete Hermann - Between Physics and Philosophy}.
\newblock Springer, 2016.

\bibitem{Dieks2017}
D.~Dieks.
\newblock {Von Neumann's impossibility proof: mathematics in the service of
  rhetorics}.
\newblock {\em Studies in History and Philosophy of Science Part B: Studies in
  History and Philosophy of Modern Physics}, 60:136--148, 2017.
\newblock On the History of the Quantum, HQ4.

\bibitem{Doering2004}
A.~{D{\"o}ring}.
\newblock {Kochen-Specker theorem for von Neumann algebras}.
\newblock {\em Int. J. Theor. Phys.}, 44:139--160, Feb. 2005.

\bibitem{Doering2008}
A.~D\"oring.
\newblock Quantum states and measures on the spectral presheaf.
\newblock {\em Adv. Sci. Lett.}, 2, 10 2008.

\bibitem{DoeringFrembs2019a}
A.~D\"oring and M.~Frembs.
\newblock Contextuality and the fundamental theorems of quantum mechanics.
\newblock {\em J. Math. Phys.}, 63(7):072103, Jul 2022.

\bibitem{DzhafarovKujalaCervantes2016}
E.~N. Dzhafarov, J.~V. Kujala, and V.~H. Cervantes.
\newblock Contextuality-by-default: A brief overview of ideas, concepts, and
  terminology.
\newblock In H.~Atmanspacher, T.~Filk, and E.~Pothos, editors, {\em Quantum
  Interaction}, pages 12--23, Cham, 2016. Springer.

\bibitem{EPR1935}
A.~Einstein, B.~Podolsky, and N.~Rosen.
\newblock Can quantum-mechanical description of physical reality be considered
  complete?
\newblock {\em Phys. Rev.}, 47:777--780, May 1935.

\bibitem{Frembs2022a}
M.~Frembs.
\newblock Bipartite entanglement and the arrow of time.
\newblock {\em Phys. Rev. A}, 107:022218, Feb 2023.

\bibitem{Frembs2024b}
M.~Frembs, 2024.
\newblock (forthcoming).

\bibitem{FrembsDoering2022a}
M.~Frembs and A.~D\"oring.
\newblock Characterization of nonsignaling bipartite correlations corresponding
  to quantum states.
\newblock {\em Phys. Rev. A}, 106:062420, Dec 2022.

\bibitem{FrembsDoering2023}
M.~Frembs and A.~D\"oring.
\newblock Gleason’s theorem for composite systems.
\newblock {\em J. Phys. A}, 56(44):445303, Oct 2023.

\bibitem{FrembsRobertsBartlett2018}
M.~Frembs, S.~Roberts, and S.~D. Bartlett.
\newblock Contextuality as a resource for measurement-based quantum computation
  beyond qubits.
\newblock {\em New J. Phys.}, 20(10):103011, Oct 2018.

\bibitem{FrembsRobertsCampbellBartlett2023}
M.~Frembs, S.~Roberts, E.~T. Campbell, and S.~D. Bartlett.
\newblock Hierarchies of resources for measurement-based quantum computation.
\newblock {\em New J. Phys.}, 25(1):013002, Jan 2023.

\bibitem{Gleason1975}
A.~M. Gleason.
\newblock {\em {Measures on the closed subspaces of a Hilbert space}}, pages
  123--133.
\newblock Springer, 1975.

\bibitem{Groenewold1946}
H.~J. Groenewold.
\newblock On the principles of elementary quantum mechanics.
\newblock {\em Physica}, 12(7):405--460, 1946.

\bibitem{GuehneEtAl2010}
O.~G\"uhne, M.~Kleinmann, A.~Cabello, J.-A. Larsson, G.~Kirchmair,
  F.~Z\"ahringer, R.~Gerritsma, and C.~F. Roos.
\newblock Compatibility and noncontextuality for sequential measurements.
\newblock {\em Phys. Rev. A}, 81:022121, Feb 2010.

\bibitem{IshamButterfieldIII}
J.~Hamilton, C.~J. Isham, and J.~Butterfield.
\newblock {Topos perspective on the Kochen-Specker theorem: III. Von Neumann
  algebras as the base category}.
\newblock {\em Int. J. Theor. Phys.}, 39(6):1413--1436, 2000.

\bibitem{HeinosaariMiyaderaZiman2016}
T.~Heinosaari, T.~Miyadera, and M.~Ziman.
\newblock An invitation to quantum incompatibility.
\newblock {\em J. Phys. A}, 49(12):123001, Feb 2016.

\bibitem{Heisenberg1927}
W.~Heisenberg.
\newblock {{\"U}ber den anschaulichen Inhalt der quantentheoretischen Kinematik
  und Mechanik}.
\newblock {\em Zeitschrift f{\"u}r Physik}, 43(3):172--198, Mar 1927.

\bibitem{stanford}
C.~Held.
\newblock {The Kochen-Specker theorem}.
\newblock In E.~N. Zalta and U.~Nodelman, editors, {\em The {Stanford}
  Encyclopedia of Philosophy}. Metaphysics Research Lab, Stanford University,
  {F}all 2022 edition, 2022.

\bibitem{HLS2009}
C.~Heunen, N.~P. Landsman, and B.~Spitters.
\newblock A topos for algebraic quantum theory.
\newblock {\em Commun. Math. Phys.}, 291(1):63--110, 2009.

\bibitem{HLS2010}
C.~Heunen, N.~P. Landsman, and B.~Spitters.
\newblock Bohrification of operator algebras and quantum logic.
\newblock {\em Synthese}, 186(3):719--752, 2012.

\bibitem{HowardEtAl2014}
M.~Howard, J.~Wallman, V.~Veitch, and J.~Emerson.
\newblock Contextuality supplies the `magic' for quantum computation.
\newblock {\em Nature}, 510(7505):351—355, June 2014.

\bibitem{IshamButterfieldI}
C.~J. Isham and J.~Butterfield.
\newblock {Topos perspective on the Kochen-Specker theorem: I. Quantum states
  as generalized valuations}.
\newblock {\em Int. J. Theor. Phys.}, 37(11):2669--2733, 1998.

\bibitem{JokinenEtAl2024}
P.~Jokinen, M.~Weilenmann, M.~Plávala, J.-P. Pellonpää, J.~Kiukas, and
  R.~Uola.
\newblock No-broadcasting characterizes operational contextuality, 2024.

\bibitem{Kernaghan1994}
M.~Kernaghan.
\newblock {Bell-Kochen-Specker theorem for 20 vectors}.
\newblock {\em J. Phys. A}, 27(21):L829, Nov 1994.

\bibitem{KernaghanPeres1995}
M.~Kernaghan and A.~Peres.
\newblock Kochen-specker theorem for eight-dimensional space.
\newblock {\em Phys. Lett., A}, 198(1):1--5, 1995.

\bibitem{KirchmairEtAl2009}
G.~Kirchmair, F.~Z{\"a}hringer, R.~Gerritsma, M.~Kleinmann, O.~G{\"u}hne,
  A.~Cabello, R.~Blatt, and C.~F. Roos.
\newblock State-independent experimental test of quantum contextuality.
\newblock {\em Nature}, 460(7254):494--497, Jul 2009.

\bibitem{KleinmannEtAl2012}
M.~Kleinmann, C.~Budroni, J.-A. Larsson, O.~G\"uhne, and A.~Cabello.
\newblock Optimal inequalities for state-independent contextuality.
\newblock {\em Phys. Rev. Lett.}, 109:250402, Dec 2012.

\bibitem{KlyachkoEtAl2008}
A.~A. Klyachko, M.~A. Can, S.~Binicio\v{g}lu, and A.~S. Shumovsky.
\newblock Simple test for hidden variables in spin-1 systems.
\newblock {\em Phys. Rev. Lett.}, 101:020403, Jul 2008.

\bibitem{KochenSpecker1967}
S.~Kochen and E.~P. Specker.
\newblock The problem of hidden variables in quantum mechanics.
\newblock {\em J. Math. Mech.}, 17:59--87, 1967.

\bibitem{Kunjwal2014}
R.~Kunjwal.
\newblock {A note on the joint measurability of POVMs and its implications for
  contextuality}, 2014.

\bibitem{LapkiewiczEtAl2011}
R.~Lapkiewicz, P.~Li, C.~Schaeff, N.~K. Langford, S.~Ramelow,
  M.~Wie{\'{s}}niak, and A.~Zeilinger.
\newblock Experimental non-classicality of an indivisible quantum system.
\newblock {\em Nature}, 474(7352):490--493, Jun 2011.

\bibitem{Larsson2002}
J.-A. Larsson.
\newblock {A Kochen-Specker inequality}.
\newblock {\em EPL}, 58(6):799, Jun 2002.

\bibitem{LeiferDuarte2020}
M.~Leifer and C.~Duarte.
\newblock Noncontextuality inequalities from antidistinguishability.
\newblock {\em Phys. Rev. A}, 101:062113, Jun 2020.

\bibitem{Marage1999}
P.~Marage and G.~Wallenborn.
\newblock {\em {The debate between Einstein and Bohr, or how to interpret
  quantum mechanics}}, pages 161--174.
\newblock Birkh{\"a}user, Basel, 1999.

\bibitem{Mermin1990}
N.~D. Mermin.
\newblock Simple unified form for the major no-hidden-variables theorems.
\newblock {\em Phys. Rev. Lett.}, 65:3373--3376, Dec 1990.

\bibitem{MoussaEtAl2010}
O.~Moussa, C.~A. Ryan, D.~G. Cory, and R.~Laflamme.
\newblock Testing contextuality on quantum ensembles with one clean qubit.
\newblock {\em Phys. Rev. Lett.}, 104:160501, Apr 2010.

\bibitem{OkayRaussendorf2020}
C.~Okay and R.~Raussendorf.
\newblock Homotopical approach to quantum contextuality.
\newblock {\em Quantum}, 4:217, Jan 2020.

\bibitem{OkayRoberts2017}
C.~Okay, S.~Roberts, S.~D. Bartlett, and R.~Raussendorf.
\newblock Topological proofs of contextuality in quantum mechanics.
\newblock {\em Quantum Inf. Comput.}, 17(13{\&}14):1135--1166, 2017.

\bibitem{Peres1991}
A.~Peres.
\newblock {Two simple proofs of the Kochen-Specker theorem}.
\newblock {\em J. Phys. A}, 24(4):L175--L178, Feb 1991.

\bibitem{Pitowsky1989}
I.~Pitowsky.
\newblock {\em From George Boole To John Bell --- The Origins of Bell's
  Inequality}, pages 37--49.
\newblock Springer Netherlands, Dordrecht, 1989.

\bibitem{Pitowsky1991}
I.~Pitowsky.
\newblock Correlation polytopes: Their geometry and complexity.
\newblock {\em Math. Program.}, 50(1):395--414, Mar 1991.

\bibitem{RamanathanHorodecki2014}
R.~Ramanathan and P.~Horodecki.
\newblock Necessary and sufficient condition for state-independent contextual
  measurement scenarios.
\newblock {\em Phys. Rev. Lett.}, 112:040404, Jan 2014.

\bibitem{RamanathanEtAl2012}
R.~Ramanathan, A.~Soeda, P.~Kurzy\ifmmode~\acute{n}\else \'{n}\fi{}ski, and
  D.~Kaszlikowski.
\newblock Generalized monogamy of contextual inequalities from the
  no-disturbance principle.
\newblock {\em Phys. Rev. Lett.}, 109:050404, Aug 2012.

\bibitem{Raussendorf2013}
R.~{Raussendorf}.
\newblock {Contextuality in measurement-based quantum computation}.
\newblock {\em Phys. Rev. A}, 88(2):022322, Aug. 2013.

\bibitem{Raussendorf2016}
R.~Raussendorf.
\newblock Cohomological framework for contextual quantum computations.
\newblock {\em ArXiv e-prints}, 2016.

\bibitem{SimonBruknerZeilinger2001}
C.~Simon, {\v{C}}.~Brukner, and A.~Zeilinger.
\newblock Hidden-variable theorems for real experiments.
\newblock {\em Phys. Rev. Lett.}, 86:4427--4430, May 2001.

\bibitem{Specker1960}
E.~Specker.
\newblock {Die Logik nicht gleichzeitig entscheidbarer Aussagen}.
\newblock {\em Dialectica}, 14:239--246, 1960.

\bibitem{Spekkens2005}
R.~W. Spekkens.
\newblock Contextuality for preparations, transformations, and unsharp
  measurements.
\newblock {\em Phys. Rev. A}, 71:052108, May 2005.

\bibitem{vanHove1951b}
L.~van Hove.
\newblock Sur certaines repr\'esentations unitaires d\'un groupe infini de
  transformations.
\newblock {\em Memoires de l\'Academie Royale Belgique}, 6(26):1--102, 1951.

\bibitem{vonNeumann1932}
J.~von Neumann.
\newblock {\em Mathematische Grundlagen der Quantenmechanik}.
\newblock {Die Grundlehren der mathematischen Wissenschaften}. Springer, 1996.

\bibitem{XuCabello2019}
Z.-P. Xu and A.~Cabello.
\newblock Necessary and sufficient condition for contextuality from
  incompatibility.
\newblock {\em Phys. Rev. A}, 99:020103, Feb 2019.

\bibitem{Zhen-PengChenGuehne2020}
Z.-P. Xu, J.-L. Chen, and O.~G\"uhne.
\newblock {Proof of the Peres conjecture for contextuality}.
\newblock {\em Phys. Rev. Lett.}, 124:230401, Jun 2020.

\bibitem{XuChenSu2015}
Z.-P. Xu, J.-L. Chen, and H.-Y. Su.
\newblock State-independent contextuality sets for a qutrit.
\newblock {\em Physics Letters A}, 379(34):1868--1870, 2015.

\bibitem{YuOh2012}
S.~Yu and C.~H. Oh.
\newblock {State-independent proof of Kochen-Specker theorem with 13 rays}.
\newblock {\em Phys. Rev. Lett.}, 108:030402, Jan 2012.

\bibitem{YuGuoTong2015}
X.-D. Yu, Y.-Q. Guo, and D.~M. Tong.
\newblock {A proof of the Kochen-Specker theorem can always be converted to a
  state-independent noncontextuality inequality}.
\newblock {\em New J. Phys.}, 17(9):093001, Sep 2015.

\bibitem{ZimbaPenrose1993}
J.~Zimba and R.~Penrose.
\newblock {On Bell non-locality without probabilities: more curious geometry}.
\newblock {\em Studies in History and Philosophy of Science Part A},
  24(5):697--720, 1993.

\end{thebibliography}

\appendix

\section{Proof of Cor.~\ref{cor: Yu-Oh for contexts}}\label{app: Yu-Oh}

The 13 vectors of the SI-C set in Ref.~\cite{YuOh2012} can be defined by their squared cosines
\begin{equation}\label{eq: 13 vectors}
    0 + 0 + 1 = 0 + \tfrac{1}{2} + \tfrac{1}{2} = \tfrac{1}{3} + \tfrac{1}{3} + \tfrac{1}{3}\; .
\end{equation}
Note that Eq.~(\ref{eq: 13 vectors}) contains pairs of orthogonal vectors, whose uniquely determined vector that completes it to an orthogonal basis (``triad") is not contained. Adding those vectors results in the set of 25 vectors, defined by their squared cosines
\begin{equation}\label{eq: 25 vectors}
    0 + 0 + 1 = 0 + \tfrac{1}{2} + \tfrac{1}{2} = \tfrac{1}{3} + \tfrac{1}{3} + \tfrac{1}{3} = \tfrac{1}{6} + \tfrac{1}{6} + \tfrac{2}{3}\; .
\end{equation}
None of these additional vectors are orthogonal, hence, the set in Eq.~(\ref{eq: 25 vectors}) is closed under forming triads.\footnote{We study ``observable algebras" generated from ``completions'' of SI-C sets in more detail in Ref.~\cite{Frembs2024b}.} We define the set of spin-$1$ observables $\cS_\mathrm{YO}$ as those spin-$1$ observables whose spectral resolution consists of rank-$1$ projections corresponding to these vectors; and $\cS_\mathrm{YO}$ generates the partial order of contexts $\cC(\cS_\mathrm{YO})\subset\CH$ in $\dim(\cH)=3$. It is easy to see that Eq.~(\ref{eq: 25 vectors}) admits a KS colouring, equivalently, it allows to consistently label one squared (spin-1 observable) in every triad of $\cS_\mathrm{YO}$ by `$S_z$'. However, it is impossible to label all elements consistently by `$S_x$',`$S_y$' and `$S_z$', as shown in Tab.~\ref{tab: Yu-Oh in context}.

\begin{table}[htb!]
    \centering
    \begin{tabular}{ccc|l|l|l}
         \toprule
         \multicolumn{3}{c|}{orthogonal triad} & further z-$\perp$ rays 
         & `$S_z$' label since & `$S_y$'/`$S_x$' label since\\
         \midrule
         $001$ & $100$ & $010$ & $110$,$\bar{1}10$ & choice of z axis & \\
         $101$ & $\bar{1}01$ & $\mathit{010}$ & $\bar{1}11$,$11\bar{1}$,$14\bar{1}$,$\bar{1}41$ & choice of x vs -x & $\perp$ to other rays\\
         $011$ & $\mathit{100}$ & $0\bar{1}1$ & $1\bar{1}1$,$41\bar{1}$,$4\bar{1}1$ & choice of y vs -y & $\perp$ to other rays\\
         $1\bar{1}4$ & $\bar{1}11$ & $110$ & & choice of x vs y & $\perp$ to $0\bar{1}1$\\
         $114$ & $\bar{1}10$ & $11\bar{1}$ & & previous z-$\perp$ rays & $\perp$ to $110$\\
         \bottomrule
    \end{tabular}
    \caption{It is impossible to consistently label spin-1 observables in the directions specified by the vectors in Eq.~(\ref{eq: 25 vectors}), using labels `$S_z$',`$S_y$' and `$S_x$', respectively. Notation: $\bar{1} = -1$, triples of integers $abc$ denote vectors $\frac{(a,b,c)^T}{||(a,b,c)||}$, e.g. $41\bar{1}$ denotes the vector $(\frac{2}{\sqrt{3}},\frac{1}{\sqrt{6}},\frac{-1}{\sqrt{6}})^T$. The left column contains orthogonal triads in $\cC(\cS_\mathrm{YO})$, italic font indicates that a vector already appeared in previous rows. The middle columns list additional vectors orthogonal to the respective coordinate axes, as defined in the first row. The right column indicates the reason for labelling the first vector in every triad by `$S_z$', using the symmetry of the arrangement. It is impossible to label all vectors by `$S_z$',`$S_y$' and `$S_x$' consistently since $11\bar{1}$ and $0\bar{1}1$ are both labelled `$S_x$' despite being orthogonal.}
    \label{tab: Yu-Oh in context}
\end{table}

The argument in Tab.~\ref{tab: Yu-Oh in context} follows the same spirit as the argument showing that no KS colouring exists for the 33 vectors in Ref.~\cite{Peres1991}, which are given by their squared cosines.
\begin{equation}\label{eq: 33 vectors}
    0 + 0 + 1 = 0 + \tfrac{1}{2} + \tfrac{1}{2} = 0 + \tfrac{1}{3} + \tfrac{2}{3} = \tfrac{1}{4} + \tfrac{1}{4} + \tfrac{1}{2}\; .
\end{equation}
We note in passing that similar to the set in Eq.~(\ref{eq: 13 vectors}), the set in Eq.~(\ref{eq: 33 vectors}) is not closed under forming triads. Indeed, its closure contains $57$ vectors and $40$ triads.

\end{document}